\documentclass{article}%
\usepackage{hyperref}%
\usepackage{amssymb}%
\usepackage{amsthm}%
\usepackage{amsmath}%
\usepackage{stmaryrd}%
\usepackage{thm-restate}%
\usepackage{microtype}%
\usepackage{mathpartir}%
\usepackage{nameref}%
\usepackage{tikz-cd}%
\usepackage{todonotes}%
\usepackage{makecell}%
\usepackage{xspace}%
\usepackage{quiver}%

\usetikzlibrary{nfold}

\usepackage{macros}%

\newtheorem{theorem}{Theorem}
\newtheorem{proposition}[theorem]{Proposition}
\newtheorem{lemma}[theorem]{Lemma}

\theoremstyle{remark}
\newtheorem{remark}{Remark}

\usepackage{mathpartir}

\makeatletter

\NewDocumentCommand \mpr@label { m }{%
  \def\@currentlabelname{#1}%
   \phantomsection 
}

\DeclareDocumentCommand \inferdef {m m m}{%
   \inferrule*[vcenter,right=#1]{#2}{#3}
   \mpr@label{\textsc{#1}}
}%

\makeatother

\newcommand{\ruleref}[1]{Rule \nameref{#1}}

\title{Generating Higher Identity Proofs in Homotopy Type Theory}
\author{Thibaut Benjamin}

\begin{document}

\maketitle

\begin{abstract}
  Finster and Mimram have defined a dependent type theory called \catt, which
  describes the structure of \(\omega\)\-categories. Types in homotopy type
  theory with their higher identity types form weak \(\omega\)\-groupoids, so
  they are in particular weak \(\omega\)\-categories. In this article, we show
  that this principle makes homotopy type theory into a model of \catt, by
  defining a translation principle that interprets an operation on the cell of
  an \(\omega\)\-category as an operation on higher identity types. We then
  illustrate how this translation allows to leverage several mechanisation
  principles that are available in \catt, to reduce the proof effort required to
  derive results about the structure of identity types, such as the existence of
  an Eckmann-Hilton cell.
\end{abstract}

\section{Introduction}

Types in Martin-Löf type theory with intentional equality are weak
\(\omega\)-groupoids~\cite{martin-lof_intuitionistic_1984}. This is one of the
foundational observation of homotopy type
theory~\cite{theunivalentfoundationsprogram_homotopy_2013} (\hott). In this
article, we work with a fragment of Martin-Löf type theory that we call \hott to
emphasise that we care about higher identity types, even though we do not use
the univalence axiom or higher inductive types. Being weak \(\omega\)-groupoids,
types in \hott are in particular weak \(\omega\)-categories, by forgetting the
preferred direction. On the other hand, Finster and
Mimram~\cite{finster_typetheoretical_2017} have introduced the dependent type
theory \catt, which is the theory of weak \(\omega\)-categories. \catt and \hott
are both dependent type theories, one is the theory of \(\omega\)-categories
while in the other one, every type is an \(\omega\)-category. This article
establishes a formal connection between the two theories.

\paragraph{Contributions.}
In this article, we define a translation scheme that allows to transport \catt
terms onto \hott terms, by viewing higher cells in \catt as higher identity
proofs in \hott. We then prove the correctness of this translation. This means
that the translation of a well-formed term in context in \catt is a well-formed
term in \hott, whose type we characterise as the translation of the type of the
original term in \catt. The translation is defined by induction on the syntax of
\catt, and the main difficulty is to translate each of the operation making the
\(\omega\)-category structure into successive applications of the \(\J\) rule in
\hott. In the terminology of Boulier, Pedrot and
Tabareau~\cite{boulier_next_2017}, we show that \hott is a syntactic model of
\catt. We also present an implementation of this translation that generates
terms in \coq and show how it can be used in combination with mechanised
reasoning available in \catt to reduce the proof effort required to generate
certain complex terms in \hott.

\paragraph{Related works.}
Weak \(\omega\)-groupoids were first defined by
Grothendieck~\cite{grothendieck_pursuing_2021} with the intent of modelling
spaces up to homotopy. The connection between types and groupoids was first
discovered by Hofmann and Streicher~\cite{hofmann_groupoid_1998}.
Lumsdaine~\cite{lumsdaine_weak_2010}, Van den Berg and
Garner\cite{vandenberg_types_2011} and Altenkirch and
Rypacek~\cite{altenkirch_syntactical_2012} then promoted this to
\(\omega\)-groupoids, showing that types with their iterated identity types are
endowed with the structure of weak \(\omega\)-groupoids. While mathematically
weak \(\omega\)-groupoids are described as a collection of cells in every
dimension that allow for partial composition operations satisfying
associativity, unitality and exchange laws up to higher cells, in Martin-Löf
type theory, the entire structure simply emerges from the \(\J\) rule. With a
careful examination of how the \(\J\) rule yields a weak \(\omega\)-groupoid
structure, Brunerie~\cite{brunerie_homotopy_2016} proposed a definition of
\(\omega\)-groupoids formulated as a dependent type theory.

Weak \(\omega\)-categories were first defined by
Batanin~\cite{batanin_monoidal_1998}. The definition was then elaborated upon by
Leinster~\cite{leinster_higher_2004}. More recently,
Maltsiniotis~\cite{maltsiniotis_grothendieck_2010} proposed an alternative
definition, which is based on Grothendieck's definition of groupoids and
modifies it by enforcing a privileged direction. This definition has been
proven equivalent to that of Batanin and Leinster by
Ara~\cite{ara_infinigroupoides_2010} and Bourke~\cite{bourke_iterated_2018}.

The definition the theory \catt by Finster and
Mimram~\cite{finster_typetheoretical_2017} was inspired by both Brunerie's
formulation of weak \(\omega\)-groupoids in dependent type theory, and
Maltsiniotis' definition of weak \(\omega\)-categories from weak
\(\omega\)-groupoids. Several works have been conducted around this theory. It
was generalised by Dean et al.~\cite{dean_computads_2024} to give an inductive
definition of computads for weak \(\omega\)-categories, it was proven by
Benjamin, Finster and Mimram~\cite{benjamin_globular_2021} that the models of
\catt are precisely Grothendieck-Maltsiniotis \(\omega\)-categories, and several
meta-operations for \catt have been proposed, such as the suspension and
opposites~\cite{benjaminHomCategoriesComputad2024}, computation of inverses and
invertibility witnesses~\cite{benjamin_invertible_2024},
functorialisation~\cite{benja2020}, allowing for mechanised reasoning
in the proof-assistant \catt.

\paragraph{Overview of the paper.}
In Section~\ref{sec:hott}, we present the fragment of Martin-Löf type theory
that we work with, which we call \hott in this article to emphasise the fact
that we are interested in the structure of higher identity types.
Section~\ref{sec:catt} is dedicated to the presentation of the dependent type
theory \catt, which relies on a simpler type theory called \gsett. In
Section~\ref{sec:gsett-to-hott}, we present our algorithm to translate terms in
\gsett into terms in \hott, and we promote this into a translation of terms in
\catt into terms in \hott in Section~\ref{sec:catt-to-hott}. Both these
translation come with proofs of correctness. Finally, in
Section~\ref{sec:implem}, we discuss the implementation aspects, and illustrate
with the example of the Eckmann-Hilton cell how leveraging the mechanisation
available in \catt allow for defining complex terms in \hott.

\paragraph{Acknowledgements.} The author wants to thank Meven Lennon-Bertrand
for his helpful guidance on the internal of \coq, making the implementation of
the plugin possible, as well as Jamie Vicary, Ioannis Markakis, Chiara Sarti and
Wilfred Offord, for the interesting discussions around the work presented in
this paper.

\section{Overview of Martin-Löf type theory}
\label{sec:hott}
This section is dedicated to the presentation of the notation we use to work
with type theory, as well as the fragment of Martin-Löf type theory that we
place ourselves in and call \hott.

\subsection{Notations for type theories}
We work with dependent type theories, made out of contexts, types, terms and
substitutions, as well as definitional equalities. Our theories are presented
with an untyped syntax with propositional judgements witnessing well-formedness.
We use the following notation for these judgements:
\[
  \begin{cases}
  \ctxrule{\Gamma} & \text{\(\Gamma\) is a valid context} \\
  \typerule{\Gamma}{A} & \text{\(A\) is a valid type in \(\Gamma\)} \\
  \termrule{\Gamma}{u}{A} & \text{\(u\) is a term of type \(A\) in \(\Gamma\)}\\
  \subrule{\Delta}{\gamma}{\Gamma}& \text{\(\gamma\) is a substitution from
                                    \(\Gamma\) to \(\Delta\)}
  \end{cases}
\]
We call pre-contexts, pre-types, pre-terms and pre-substitutions the elements of
the untyped syntax, which are not necessarily well-formed, to contrast with
contexts, types terms and substitution which are implicitly assumed to be
well-formed. We denote \(u=v\) the syntactic equality between pre-terms \(u\)
and \(v\), and \(u\equiv v\) the definitional equality between them. Syntactic
equality compares the expressions whereas definitional equality can be subject
to equality rules. The empty context is denoted \(\emptycontext\), and
substitutions are written as \(\sub{x\mapsto u}\) where \(x\) is a variable and
\(u\) is the term the variable \(x\) is substituted for. In the rest of this
article, we introduce a three dependent type theories with which we work and we
assume that they enjoy all of the usual structure, making their syntax into a
category with families. This has already been shown for each of these, and more
details regarding this structure is given in Appendix~\ref{app:cwf}

\subsection{Martin-Löf type theory}
In this article, we consider a fragment of Martin-Löf type theory with universes
a la Tarski, identity types and \(\Pi\)-types. Since we are interested in the
structure of higher identity types, we call this theory homotopy type theory, or
\hott, even though our work does not require the univalence principle or higher
inductive types. For the sake of clarity, we avoid the use of de Bruijn level
and assume an countably infinite set of variables, that we denote
\(\hvar{x}, \hvar{y}\), \ldots We use the standard notations for the
\(\Pi\)-types as \(\Pi(x:A).B\) and \(\lambda\)-applications as
\(\lambda(x:A).u\).

Our construction does not involve size issues, so we may assume that we are
working with a universe of fixed size \(\Type\). The type constructor \(\El\)
produces a type out of a term of type \(\Type\), as given by the following rule
\[
  \inferdef{(\(\El\)-in)}{\termrule{\Gamma}{A}{\Type}}{\typerule{\Gamma}{\El(A)}}
  \label{rule:el}
\]
Given a type \(A\) and terms \(u,v\), we denote \(\Id_{A}(u,v)\) the identity
type, and given just \(A\) and \(u\), we denote \(\refl_{A,u}\) the reflexivity
term. The introductions rules for identity types and reflexevity are the
following:
\begin{mathpar}
  \inferdef{(\(\Id\)-in)}{\termrule{\Gamma}{u}{A} \\
  \termrule{\Gamma}{v}{A}}{\typerule{\Gamma}{\Id_{A}(u,v)}}\label{rule:id}
  \and
  \inferdef{(\(\refl\)-in)}{\termrule{\Gamma}{u}{A}}
  {\termrule{\Gamma}{\refl_{A,u}}{\Id_{A}(u,u)}}\label{rule:refl}
\end{mathpar}
We also use the Paulin-Mohring formulation of the \(\J\) constructor for
identity types~\cite{paulin-mohring_inductive_1993}, whose introduction rule is
given as follows:
\[
  \inferdef{(\(\J\)-in)}{
    \typerule{\Gamma}{A} \\
    \termrule{\Gamma}{u}{A}\\
    \typerule{\Gamma,x:A,y:A,e:\Id_{A}(x,y)}{P(x,y,e)} \\
    \termrule{\Gamma}{p}{P(u,u,\refl_{A,u})} }
  {\termrule{\Gamma}{\J(A,P,p)}{\Pi(v:A)(e:\Id_{A}(u,v)).P(u,v,e)}}\label{rule:j}
\]
The \(\J\) rule is subject to a computation rule. For the sake of simplicity, we
write this as the following linear untyped defintional equality:
\begin{equation}\tag{\(\eta_\J\)}
  \label{eq:eta-j}
  \J(A,u,P,p)\ v\ \refl_{B,w} \equiv p.
\end{equation}
Any application of \eqref{eq:eta-j} in a context where both sides are well-typed
must be instantiated with \(u,v,w\), as well as \(A,B\) being definitionally
equal, and thus coincide with the standard computation rule.

For any term \(u\) and any type \(A\) of \hott we define the \(\N\)-indexed
family of types \(\Id^{n}_{A}(u)\) and the \(\N\)-indexed family of terms
\(\refl^{n}_{A,u}\) by induction on \(n\) as follows:
\begin{align*}
  \Id^{0}_{A}(u) & \defeq A
  & \refl^{0}_{A,u} & \defeq u \\
  \Id^{n+1}_{A}(u)& \defeq \Id_{\Id^{n}_{A}(u)}(\refl^{n}_{A,u},\refl^{n}_{A,u})
  & \refl^{n+1}_{A,u} & \defeq \refl_{\Id^{n}_{A,u},\refl^{n}_{A,u}}
\end{align*}
These define valid types and terms in the theory \hott, in the following sense:
\begin{lemma}\label{lemma:iterated-refl}
  For every \(n\in \N\), the following rules are admissible in \hott:
  \begin{align*}
    \inferrule{\termrule{\Gamma}{u}{A}}{\typerule{\Gamma}{\Id^{n}_{A}(u)}}
    && \inferrule{\termrule{\Gamma}{u}{A}}
       {\termrule{\Gamma}{\refl^{n}_{A,u}}{\Id^{n}_{A,u}}}
  \end{align*}
\end{lemma}
\begin{proof}
  We prove this result by induction. The validity of the type \(\Id^{0}_{A}(u)\)
  is a standard inversion lemma, and the typing of the term \(\refl^{0}_{A,u}\)
  is exactly given by the premise of the rule. The validity of the type
  \(\Id^{n+1}_{A}(u)\) is obtained from \ruleref{rule:id} and the typing of the
  term \(\refl^{n}_{A,u}\). The typing of the term \(\refl^{n+1}_{A,u}\) is
  obtained from \ruleref{rule:refl} and the typing of the term
  \(\refl^{n}_{A,u}\).
\end{proof}

Given a valid type \(\typerule{\Gamma}{A}\) in \hott, we define its
\emph{\(\Pi\)-lifting} \(\Pi_{\Gamma}.A\) to be the iterated \(\Pi\)-type with
body \(A\) where all variables of \(\Gamma\) are successively bound by a \(\Pi\)
constructor. Similarly, given a valid term \(\termrule{\Gamma}{u}{A}\), we
define its \emph{\(\lambda\)-lifting} \(\lambda_{\Gamma}.u\) to be the iterated
\(\lambda\)-term with body \(u\) where all variables of \(\Gamma\) are
successively bound by a \(\lambda\)-constructor.
\begin{lemma}\label{lemma:lambda-lifting}
  The following rules are admissible in \hott:
  \begin{align*}
    \inferrule{\typerule{\Gamma}{A}}{\typerule{\emptycontext}{\Pi_{\Gamma}.A}}
    && \inferrule{\termrule{\Gamma}{u}{A}}
       {\termrule{\emptycontext}{\lambda_{\Gamma}.u}{\Pi_{\Gamma}.A}}
  \end{align*}
\end{lemma}
\begin{proof}
  The proof is a straightforward induction on the length of the context
  \(\Gamma\), applying successively the introduction rules for \(\Pi\)-types and
  \(\lambda\)-terms.
\end{proof}

\section{The dependent type theory CaTT}
\label{sec:catt}
In this section, we present the dependent types theories \gsett and \catt. These
are type-theoretic formulations of algebraic structures, namely globular sets
and weak \(\omega\)\nobreakdash-categories, using Cartmell's point of view of
dependent type theories as generalised algebraic
theories~\cite{cartmell_generalised_1986}. These type theories do not support
any of the usual constructions that one typically encounters with Martin-Löf
type theory, in particular they do not have universes, function types,
\(\Pi\)-types or \(\Sigma\)-types. The type theories \gsett and \catt have the
same type constructors, but \gsett has no term constructor, so that the only
terms in \gsett are variables, and \catt is an extension of \gsett with term
constructors.

\subsection{The type theory GSeTT}
The type theories \gsett and \catt have two type constructors \(\obj\) and
\(\arr[A]{u}{v}\) where \(A\) is a pre-type and \(u,v\) are pre-terms. The
introduction rules state that \(\obj\) is always a valid type, while
\(\arr[A]{u}{v}\) is subject to the same conditions as \(\Id_{A}(u,v)\).
Formally the rules are the following:
\begin{mathpar}
  \inferdef{(\(\obj\)-in)}{\ctxrule{\Gamma}}{\typerule{\Gamma}{\obj}}
  \label{rule:obj}
  \and
  \inferdef{(\(\arr{}{}\)-in)} {\termrule{\Gamma}{u}{A} \\
    \termrule{\Gamma}{v}{A}} {\typerule{\Gamma}{\arr[A]{u}{v}}}
  \label{rule:arr}
\end{mathpar}
We define by induction the dimension of a pre-type of \gsett or \catt as
follows:
\begin{align*}
  \dim \obj &= -1 & \dim \arr[A]uv = \dim A + 1
\end{align*}
By convention, the initialisation is at \(-1\), to be consistent with notations
used in topology. The dependent type theory \gsett is the theory with those two
type constructor, and no term constructors.

\paragraph{Intuition on the semantics.}
Contexts in the theory \gsett can be understood as a description of a structure
called globular sets, which is a generalisation of a graph, allowing
\(2\)-dimensional arrows between the arrows, \(3\)-dimensional arrows between
the \(2\)-dimensional ones, and so on. In globular sets, \(n\)-dimensional
arrows are called \(n\)-cells. A context in \gsett describes a finite globular
sets, by interpreting variables of type \(\obj\) as vertices, and variables of
an arrow type as arrows between the prescribed source and target. The following
illustrates an a few contexts and their corresponding globular sets, using a
pictorial representation of those:
\[
  \begin{array}{c@{\qquad\qquad}c}
    \left(
    \makecell{
    x:\obj,\\
    f : \arr[\obj]{x}{x}}
    \right)
    &
      \left(
      \makecell{
      x:\obj,y:\obj,z:\obj, \\
    f : \arr[\obj]{x}{y}, f':\arr[\obj]{x}{y},
    g:\arr[\obj]{z}{y}, \\
    \alpha:\arr[{\arr[\obj]{x}{y}}]{f}{f'}}
    \right)\\
    \begin{tikzcd}[ampersand replacement=\&]
      x\ar[loop,"f"']
    \end{tikzcd}
    &
      \begin{tikzcd}[ampersand replacement=\&]
        x
        \ar[r,"f", bend left]
        \ar[r, "f'"', bend right]
        \ar[r, phantom, "\Downarrow_{\alpha}"]
        \& y
        \& z \ar[l,"g"']
      \end{tikzcd}
  \end{array}
\]
A judgement \(\termrule{\Gamma}{\cvar x}{A}\) then corresponds to a choice of a
particular cell in the globular cell corresponding to \(\Gamma\), together with
its iterated sources and targets. Adopting the presheaf view on globular sets,
this is equivalent by the Yoneda lemma to a morphism of globular sets from
\(D^n\) to the globular set corresponding to \(\Gamma\), where \(D^{n}\) is the
walking \(n\)-cell globular set, also called \(n\)-dimensional disk. A more
in-depth exploration of this facts can be found in the works that directly deals
with the semantics of \catt~\cite{finster_typetheoretical_2017,
  benjamin_globular_2021}.

\subsection{The type theory CaTT}
The type theory \catt extends \gsett by adding term constructors. Intuitively
these term constructors correspond to adding the higher categories operations to
the globular sets described by the theory \gsett.

\paragraph{Ps-contexts and their source and target.}
The term constructors of the theory \catt are parameterised by a class of
contexts of \gsett that we call \emph{ps-contexts}. Those contexts are
recognised by a judgement denoted \(\psrule{\Gamma}\), which is itself dependent
on an auxiliary judgement \(\psauxrule{\Gamma}{x}{A}\) where \(x\) is a
variable and \(A\) a type in \gsett. The derivation rules for these judgements
are the following:
\begin{mathpar}
   \inferdef{(pss)}{\null}{\psauxrule{\Gamma}{x}{\obj}}
    \and
      \inferdef{(pse)}{\psauxrule{\Gamma}{x}{A}}
      {\psauxrule{\Gamma,y:A,f:\arr[A]{x}{y}}{f}{\arr[a]xy}} \\
    \inferdef{(psd)}{\psauxrule{\Gamma}{f}{\arr[A]xy}}{\psauxrule{\Gamma}{y}{A}}
    \and \inferdef{(ps)}{\psauxrule{\Gamma}{x}{\obj}}{\psrule{\Gamma}}
\end{mathpar}
Additionally, we also require that all variables bound by ps-contexts are
different. Intuitively, ps-contexts are the contexts of \gsett that do not
present holes, and have a global directionality.
\begin{lemma}
  \label{lemma:admissibility-ps}
  Ps-contexts are valid \gsett contexts. More precisely, the following rules are
  admissible in \gsett
  \begin{align*}
    \inferrule{\psrule{\Gamma}}{\ctxrule{\Gamma}}
    &&
       \inferrule{\psauxrule{\Gamma}{x}{A}}{\termrule{\Gamma}{x}{A}}
  \end{align*}
\end{lemma}
This result proven in~\cite{benjamin_globular_2021} shows that ps-context
correspond to a particular class of globular sets. They can be recognised by the
existence of a total order on their cells~\cite{weber_generic_2004,
  finster_typetheoretical_2017}, and correspond to situations that can be
contracted in a directed sense. For instance the following contexts are
ps-contexts (illustrated here with their corresponding globular set):
\[
  \begin{array}{c@{\qquad\qquad\qquad}c}
    \left(
    \makecell[l]{
    x:\obj,y:\obj,f:\arr[\obj]{x}{y},\\
    z:\obj,g:\arr[\obj]{y}{z}
    }
    \right)
    &
      W\defeq
      \left(
      \makecell[l]{
      x:\obj,y:\obj,f:\arr[\obj]{x}{y},\\
    g:\arr[\obj]{x}{y},a:\arr[{\arr[\obj]{x}{y}}]{f}{g},\\
    z:\obj,h:\arr[\obj]{y}{z}
    }
    \right)
    \\
    \begin{tikzcd}[ampersand replacement=\&]
      x
      \ar[r, "f"]
      \& y \ar[r, "g"]
      \& z
    \end{tikzcd}
    &
    \begin{tikzcd}[ampersand replacement=\&]
      x
      \ar[r, bend left, "f"]
      \ar[r, bend right, "g"']
      \ar[r, phantom, "\Downarrow_{a}"]
      \& y \ar[r, "h"]
      \& z
    \end{tikzcd}
  \end{array}
\]
Given a ps-context \(\Gamma\) and an integer \(i\), we define two sub-contexts
\(\partial_{i}^{-}\Gamma\) and \(\partial_{i}^{+}\Gamma\) called respectively
the \(i\)-source and the \(i\)-target of \(\Gamma\) as follows:
\begin{align*}
  \partial_{i}^{-}(x:\obj)
  & \defeq (x:\obj)\\
  \partial_{i}^{-}(\Gamma,y:A,f:\arr[A]{x}{y})
  & \defeq
    \begin{cases}
      \partial_{i}^{-}\Gamma & \text{if \(\dim A + 1 \geq i\)} \\
      (\partial_{i}^{-}\Gamma,y:A,f:\arr[A]xy) & \text{otherwise}
    \end{cases} \\
  \partial_{i}^{+}(x:\obj)
  & \defeq (x:\obj)\\
  \partial_{i}^{+}(\Gamma,y:A,f:\arr[A]{x}{y})
  & \defeq
    \begin{cases}
      \partial_{i}^{-}\Gamma & \text{if \(\dim A + 1 > i\)} \\
      (\tail(\partial_{i}^{-}\Gamma),y:A) & \text{if \(\dim A + 1 > i\)} \\
      (\partial_{i}^{-}\Gamma,y:A,f:\arr[A]xy) & \text{otherwise}
    \end{cases}
\end{align*}
where \(\tail\) is the operator that removes the top first element of the list.
This definition illustrates an important principle when reasoning on
ps-contexts, only allow to start with a one-element list and add two elements at
one time, performing structural induction on a derivation that a context is a
pasting scheme amounts to reasoning by induction on the pasting scheme seen as
an odd-sized list. This is because there exists at most one derivation than a
given context is a ps-context~\cite{benjamin_formalization_2021}. We illustrate
the sources and targets of the ps-context \(W\) introduced above:
\[
  \begin{array}{c@{\qquad}c@{\qquad}c@{\qquad}c}
     \partial^{-}_{1}W
    & \partial^{+}_{1}W
    & \partial^{-}_{0}W
    & \partial^{+}_{0}W \\
      \left(
      \makecell[l]{
      x:\obj,y:\obj,f:\arr[\obj]{x}{y},\\
    z:\obj,h:\arr[\obj]{y}{z}
    }
    \right)
    &
      \left(
      \makecell[l]{
      x:\obj,y:\obj,g:\arr[\obj]{x}{y},\\
    z:\obj,h:\arr[\obj]{y}{z}
    }
    \right)
    &
      \left(
    \makecell[l]{x:\obj}
    \right)
    &
      \left(
      \makecell[l]{
      z:\obj
    }
      \right)
    \\
    \begin{tikzcd}[ampersand replacement=\&]
      x
      \ar[r, bend left, "f"]
      \& y \ar[r, "h"]
      \& z
    \end{tikzcd}
    &
      \begin{tikzcd}[ampersand replacement=\&]
        x
        \ar[r, bend right, "g"']
        \& y \ar[r, "h"]
        \& z
      \end{tikzcd}
    &
      \begin{tikzcd}[ampersand replacement=\&]
        x
      \end{tikzcd}
    &
      \begin{tikzcd}[ampersand replacement=\&]
        z
      \end{tikzcd}
  \end{array}
\]
For any \(i\geq 2\), the \(i\)-sources and \(i\)-targets of \(W\) are equal to
\(W\). This is a generic fact: Given a ps-context \(\Gamma\) and
\(i\geq \dim \Gamma\), the \(i\)-sources and targets of \(\Gamma\) are equal to
\(\Gamma\).

\paragraph{Term constructors.}
The theory \catt has a family of term constructors \(\coh_{\Gamma,A}\), where
\(\Gamma\) is a ps-context and \(A\) is a type. Not all pairs \((\Gamma,A)\) of
contexts and types yield a valid term constructor, so we introduce a judgement
\(\cohrule{\Gamma}{A}\) to signify when this is the case. We sometimes call the
\emph{coherences} the valid term constructors \(\cohrule{\Gamma}{A}\), and the
following derivation rule characterises these coherences, for
\(i\in\{\dim \Gamma-1,\dim \Gamma\}\):
\begin{mathpar}
  \inferdef{(\(\coh\)-wd)}{
  \psrule{\Gamma} \\
  \termrule{\partial_{i}^{-}\Gamma}{u}{A} \\
  \termrule{\partial_{i}^{+}\Gamma}{v}{A}
  }
  {\cohrule{\Gamma}{\arr[A]{u}{v}}}
  \label{rule:coh-wd}
     \begin{cases}
       \Var(u:A) &= \Var(\partial_{i}^{-}\Gamma) \\
       \Var(v:A) &= \Var(\partial_{i}^{+}\Gamma).
     \end{cases}
\end{mathpar}
Here, \(\Var(u:A)\) denotes the union of all variables used by the term \(u\)
and the type \(A\), and \(\Var \Gamma\) denotes all variables bound in context
\(\Gamma\). The term constructors are subject to the following introduction
rule:
\[
  \inferdef{(\(\coh\)-in)}{\cohrule{\Gamma}{A} \\
    \subrule{\Delta}{\gamma}{\Gamma}}
  {\termrule{\Delta}{\coh_{\Gamma,A}[\gamma]}{A[\gamma]}}
  \label{rule:coh}
\]

\begin{remark}
  Some presentations of \catt use two different term constructors or two
  different introduction rules for \(\coh\), and sometimes \ruleref{rule:coh-wd}
  and \ruleref{rule:coh} are combined into a single rule. It is straightforward
  that those are equivalent, and the presentation we give here is better suited
  for the work in this article.
\end{remark}

\subsection{Implementation and intuition on the semantics}
We have implemented in \ocaml a type checker for the theory \catt, that we also
call \catt\footnote{\url{https://www.github.com/thibautbenjamin/catt}}. This plays the
role of a proof assistant specifically dedicated to working with a mathematical
structure called \(\omega\)-categories~\cite{maltsiniotis_grothendieck_2010}. In
this proof assistant the user may declare the coherence \(\coh_{\Gamma,A}\) with
the following syntax
\begin{verbatim}
coh [name] [ps-context] : [type-expr]
\end{verbatim}
and define the term \(\termrule{\Gamma}{u}{A}\) using one of the following
syntax
\begin{verbatim}
let [name] [context] = [body]
let [name] [context] : [type-expr] = [body]
\end{verbatim}
where \verb|[name]| is a string representing a name given to the coherence of
the term, \verb|[context]| and \verb|[ps-context]| are description of the
context \(\Gamma\), \verb|type-expr| is the type expression representing \(A\)
and \verb|body| is the term expression representing \(t\). For instance,
consider the following ps-context (represented with the corresponding diagram)
and coherence (where the output type), representing the composition of
\(1\)-cells
\[
  \begin{array}{c@{\qquad\qquad}c}
    \Gamma_{2} \defeq \left(
    \makecell{
    x:\obj,y:\obj,f:\arr[\obj]xy, \\
    \phantom{x:\obj,} z:\obj, g:\arr[\obj]yz}
    \right)
    &
      \comp\defeq \coh_{\Gamma_{2},\arr[\obj]xz}
  \end{array}
\]
This coherence can be declared with the name \verb|comp| in the proof-assistant
\catt using the following input:
\begin{verbatim}
coh comp (x y : *) (f : x -> y) (z : *) (g : y -> z) : x -> z
\end{verbatim}
In the proof arguments, contexts are declared as list of parenthesised items, and
in the type \(\arr[A]uv\), the type \(A\) is always left implicit and inferred
from \(u\) and \(v\). Consider now the following context together with the term,
representing the way to compose three \(1\)-cells together by first composing
the first two then composing the result with the last cell
\[
  \begin{array}{c@{\qquad\qquad}c}
    \Gamma_{3} \defeq \left(
    \makecell{
    x:\obj,y:\obj,f:\arr[\obj]xy, \\
    \phantom{x:\obj,} z:\obj, g:\arr[\obj]yz\\
    \phantom{x:\obj,} w:\obj, g:\arr[\obj]zw
    }
    \right)
    &
    \termrule{\Gamma_{3}}{\lcomp}{\arr[\obj]xw}
    \\[1.5em]
    \multicolumn{2}{c}{
    \lcomp\defeq \comp\left[\left\langle
    \makecell[l]{x\mapsto x, y\mapsto z, f \mapsto \comp\left[\left\langle
    \makecell[l]{x\mapsto x, y\mapsto y, f\mapsto f,\\
    {\phantom{x\mapsto x,}} z\mapsto z, g\mapsto g}
    \right\rangle\right],\\
    \phantom{x\mapsto x,} z\mapsto w, g\mapsto h}
    \right\rangle\right]
    }
  \end{array}
\]
This term can be declared in the proof assistant \catt with the following input:
\begin{verbatim}
let lcomp (x y z w : *) (f : x -> y) (g : y -> z) (h : z -> w)
    = comp (comp f g) h
\end{verbatim}
In this declaration, since the context need not be a ps-context, we grouped
together all variables of type \(\obj\) for simplification. This also
illustrates that substitutions are written like applications, and that a lot of
arguments are left implicit. The proof assistant \catt can automatically infer
which arguments should be implicit, matching the usual notational conventions
where the domains and codomains of functions are implicit. As hinted by these
examples, terms in \catt correspond to operations on weak \(\omega\)-categories.
Coherences represent primitive operations while terms are cells obtained by
using those primitive operations. Here we use the word operation in a liberal
sense. Consider the following coherence
\begin{verbatim}
coh assoc (x : * ) (y : *) (f : x -> y)
                   (z : *) (g : x -> y)
                   (w : *) (h : y -> z)
    : comp (comp f g) h -> comp f (comp g h)
\end{verbatim}
This coherence define a cell known as the associator, which witnesses that the
composition is associative, but in a weak sense. That is, the two terms
\verb|comp (comp f g) h| and \verb|comp f (comp g h)| are not definitionally
equal, but related by this higher cell, which happens to be invertible in some
appropriate sense~\cite{cheng_ocategory_2007, fujii_oweak_2024,
  benjamin_invertible_2024}.

A formal connection between the semantics of \catt and the theory of weak
\(\omega\)-categories has been established~\cite{benjamin_globular_2021},
showing that \catt allows one to work directly in the theory of weak
\(\omega\)-categories. Since types in \hott are weak
\(\omega\)-groupoids~\cite{lumsdaine_weak_2010, vandenberg_types_2011,
  altenkirch_syntactical_2012}, they are in particular weak
\(\omega\)-categories. The next two section are dedicated to expanding in the
connection between \hott and \catt, by showing that \hott is a syntactic model
of \catt, that is, by defining a translation from the syntax of \hott to that of
\catt and showing that this translation preserves the typing judgements.

\section{Translation of GSeTT terms into HoTT}
\label{sec:gsett-to-hott}
Before presenting the translation of \catt terms into \hott, we present that of
\gsett terms into \hott. Not only this presentation is useful as a buildup
towards the translation of \catt terms in a simpler setting, it is also directly
used by that translation, as the coherence rule of \catt terms depend on
derivation in \gsett through ps-contexts. From now on, in order to simplify the
notations, we assume the following on variables:
\begin{itemize}
\item Variables names in \gsett and in \catt are the same
\item Every variable name in \catt or \gsett is also a variable name in \hott
\item There are variable names in \hott that are not valid names in \catt.
\end{itemize}
These assumptions are possible: if we assume both sets of variables to be
countably infinite, there exists an non-bijective injection from variable names
in \catt to variables names of \hott. By convention we use the font \(\cvar{x}\)
to denote a variable in \gsett or \catt and the font \(\hvar{x}\) to denote a
variable in \hott, and we use the name of the variable to keep track of the
names.

We consider a term judgement \(\termrule{\Gamma}{\cvar{x}}{A}\) in \gsett. Our
translation interprets the variable \(\cvar{x}\) as a closed term representing
the projection of the interpretation of the context \(\Gamma\) onto the variable
\(\hvar{x}\). The context \(\Gamma\) is understood as describing the arguments
of the function, over a generic type \(\hvar{B}\), where variables of type
\(\obj\) represent arguments of type \(\hvar{B}\) and variables of type
\(\arr[]{\cvar{y}}{\cvar{z}}\) represent equalities between the corresponding
variables. This translation is best understood by considering a few simple
examples. Figure~\ref{fig:translation-vars} illustrates it on a few \gsett
terms, displayed with their corresponding closed \(\lambda\)-terms.
\begin{figure}
  \centering
  \(
  \begin{array}{|c|c|}
    \hline
    \text{\gsett term in context}
    & \text{closed \(\lambda\)-term} \\
    \hline
    \termrule{\cvar{x} : \obj}{\cvar{x}}{\obj}
    & \lambda (\hvar{B}:\Type) ({\hvar{x}}:\hvar{B}). x\\
    \termrule{\cvar{x} : \obj,\cvar{f}:\arr{\cvar{x}}{\cvar{x}}}{\cvar{f}}{\arr{\cvar{x}}{\cvar{x}}}
    & \lambda (\hvar{B}:\Type) (\hvar{x}:\hvar{B})(\hvar{f} : \Id_{\hvar{B}}(\hvar{x},\hvar{x})). \hvar{f}\\
    \termrule{\cvar{x}:\obj,\cvar{y}:\obj,\cvar{f}:\arr{\cvar{x}}{\cvar{y}}}{\cvar{y}}{\obj}
    & \lambda (\hvar{B}:\Type) (\hvar{x}:\hvar{B})(\hvar{y}:\hvar{B})(\hvar{f} : \Id_{\hvar{B}}(\hvar{x},\hvar{y})). \hvar{y}\\
    \hline
  \end{array}
  \)
  \caption{Translation from \gsett terms to \hott}
  \label{fig:translation-vars}
\end{figure}
As illustrated by this figure, when no term constructor is involved, all the
functions resulting from the translation are projections, they just return one
of their arguments unchanged and discard all the other arguments. Although this
case does not produce interesting functions, it lets us illustrates a few
principles about the translation. First, it shows how the type of objects
\(\obj\) is translated into a reference to a universally quantified type
\(\hvar{B}\), so the functions resulting from the translations must take
\(\hvar{B}\) as their first argument. Secondly, it shows how a \gsett context
can be understood as describing an sequence of arguments that are all of type
\(\hvar{B}\) or of an higher identity type above those.

To define the translation, we first fix a variable \(\hvar{B}\) in \hott whose
not is not a valid name in \gsett, and define translation functions
\(\interp{\_}_{\hvar{B}}\), which associate to a context (resp.~a type) of
\catt, a context (resp.~a type) in \hott. These function are defined by
induction as follows:
\begin{align*}
  \begin{cases}
    \interp{\emptycontext}_{\hvar{B}} & \defeq (\hvar{B}:\Type) \\
    \interp{(\Gamma.\cvar{x}:A)}_{\hvar{B}} & \defeq (\interp{\Gamma}_{\hvar{B}}.\hvar{x}:\interp{A}_{\hvar{B}})
  \end{cases}
  &&
    \begin{cases}
      \interp{\obj}_{\hvar{B}} & \defeq \El(\hvar{B}) \\
      \interp{\arr[A]{\cvar{x}}{\cvar{y}}}_{\hvar{B}} & \defeq \Id_{\interp{A}_{\hvar{B}}}(\hvar{x},\hvar{y})
    \end{cases}
\end{align*}
\begin{restatable}{lemma}{translationgsett}
\label{lemma:translation-gsett}
  Translating derivable contexts, types and variables in \gsett yields derivable
  contexts, types and variables in \hott. More precisely, we have:
  \begin{itemize}
  \item For every derivable context judgement \(\ctxrule{\Gamma}\) in \gsett,
    the judgement \(\ctxrule{\interp{\Gamma}_{\hvar B}}\) is derivable in \hott.
  \item For any type judgement \(\typerule{\Gamma}{A}\) derivable in \gsett, the
    judgement \(\typerule{\interp{\Gamma}_{\hvar{B}}}{\interp{A}_{\hvar{B}}}\) is
    derivable in \hott.
  \item For any term judgement \(\termrule{\Gamma}{\cvar{x}}{A}\) derivable in
    \gsett, the judgement
    \(\termrule{\interp{\Gamma}_{\hvar{B}}}{\hvar{x}}{\interp{A}_{\hvar{B}}}\)
    is derivable in \hott
  \end{itemize}
\end{restatable}
\begin{proof}
  This can be proved by induction on the derivation tree of the judgement in the
  theory \gsett. See Appendix~\ref{app:gsett-to-hott} for a complete proof.
\end{proof}

Given a derivable type judgement \(\typerule{\Gamma}{A}\) in \gsett, we define
the closed type \(\interp{\typerule{\Gamma}{A}}\) as the \(\Pi\)-lifting of the
type judgement \(\typerule{\interp{\Gamma}_{\hvar{B}}}{\interp{A}_{\hvar{B}}}\)
in \hott. Similarly, for a derivable term judgement
\(\termrule{\Gamma}{\cvar x}{A}\) in \gsett, we define the closed term
\(\interp{\termrule{\Gamma}{\cvar{x}}{A}}\) to be the \(\lambda\)-lifting of the
derivable term judgement
\(\termrule{\interp{\Gamma}_{\hvar{B}}}{\hvar{x}}{\interp{A}_{\hvar{B}}}\) in
\hott.
\begin{align*}
  \interp{\typerule{\Gamma}{A}}
  & \defeq \Pi_{\interp{\Gamma}_{\hvar B}}.\interp{A}_{\hvar B} &
  \interp{\termrule{\Gamma}{\cvar x}{A}}
  & \defeq \lambda_{\interp{\Gamma}_{B}}.\hvar x
\end{align*}
The choice of the variable \(\hvar{B}\) is irrelevant as any two choices lead to
\(\alpha\)-equivalent terms, hence the variable disappears from the notation.
\begin{proposition}
  Consider a judgement \(\termrule{\Gamma}{\cvar{x}}{A}\) derivable in \catt
  where no term constructor appears. Then the closed \(\lambda\)-term
  \(\interp{\termrule{\Gamma}{\cvar{x}}{A}}\) is valid in \hott. More precisely,
  the following judgement is derivable in \hott:
  \[
    \termrule{\emptycontext} {\interp{\termrule{\Gamma}{\cvar{x}}{A}}}
    {\interp{\typerule{\Gamma}{A}}}.
  \]
\end{proposition}
\begin{proof}
  This is a consequence of Lemmas~\ref{lemma:lambda-lifting}
  and~\ref{lemma:translation-gsett}
\end{proof}

\section{Translation of CaTT terms into HoTT}
\label{sec:catt-to-hott}
In this section, we extend our translation to a translation from \catt term to
closed \(\lambda\)-term in \hott. Intuitively, a \catt term in a context is
interpreted in \hott as a function taking arguments as described by the \catt
context, and returning a value determined by the term. As in the previous
section, we start by defining functions that perform the translation in
contexts, and then we take the \(\lambda\)-lifting. From now on, we consider a
variable \(\hvar B\) that is a variable in \hott which is not a name in \catt.
The translation functions \(\interp{\_}_{\hvar{B}}\) are now more numerous and
more complex. Not only do we consider translations of contexts and types, but
also of terms, substitutions and coherences. The translation function on
coherences sends a coherence of \catt onto a closed term in \(\hott\). These
functions are all defined by mutual induction. The cases for the contexts,
types, terms and substitutions are straightforward by induction and given by the
following:
\begin{align*}
  & \begin{cases}
    \interp{\emptycontext}_{\hvar{B}}
    & \defeq (\hvar{B}:\Type) \\
    \interp{(\Gamma.\cvar{x}:A)}_{\hvar{B}}
    & \defeq (\interp{\Gamma}_{\hvar{B}}.\hvar{x}:\interp{A}_{\hvar B})
  \end{cases}
  &&
    \begin{cases}
      \interp{\obj}_{\hvar{B}}
      & \defeq \El(\hvar{B}) \\
      \interp{\arr[A]{t}{u}}_{\hvar{B}}
      & \defeq \Id_{\interp{A}_{\hvar{B}}}(\interp{t}_{\hvar{B}},\interp{u}_{\hvar{B}})
    \end{cases} \\
  &\begin{cases}
    \interp{\cvar{x}}_{\hvar{B}}
    & \defeq \hvar{x} \\
    \interp{\coh_{\Gamma,A}[\gamma]}_{\hvar{B}}
    & \defeq (\interp{\coh_{\Gamma,A}}_{\hvar{B}})\left[\interp{\gamma}_{\hvar{B}}\right]
  \end{cases}
  &&
     \begin{cases}
       \interp{\sub{}}_{\hvar{B}}
       & \defeq \sub{\hvar{B}\mapsto \hvar{B}} \\
       \interp{\sub{\gamma,\cvar{x}\mapsto t}}_{\hvar{B}}
       & \defeq \sub{\interp{\gamma}_{\hvar{B}},\hvar{x}\mapsto\interp{t}_{\hvar{B}}}.
     \end{cases}
\end{align*}
The definition of \(\interp{\coh_{\Gamma,A}}_{\hvar{B}}\) is left in order to
complete the mutual induction. We define this case as
\(\interp{\coh_{\Gamma,A}}_{\hvar B}\defeq \elim_{\hvar
  B}(\Gamma,A,\Gamma,\id_{\interp{\Gamma}_{\hvar B}})\), where \(\elim\) is the
auxiliary function defined mutually inductively, by induction on the derivation
of the ps-context:
\[
  \begin{cases}
    \elim_{\hvar B}((\cvar x : \obj),A,\Gamma,\gamma)
    & \defeq \refl_{\hvar B,\hvar x}^{\dim A + 1} \\
    \elim_{\hvar B}((\Delta,\cvar y:C,\cvar{f}:\arr[C]{\cvar x}{\cvar y}), A,\Gamma,\gamma)
    & \defeq
      \J({\interp{C}_{\hvar B}},
      \hvar x,
      P(\hvar{x'},\hvar{y'},\hvar{f'}),
      p)\
      \hvar y\
      \hvar f
  \end{cases}
\]
where \(P\) and \(p\) are respectively defined as follows
\begin{align*}
  P(\hvar{x'},\hvar{y'},\hvar{f'})
  & \defeq \interp{A}_{\hvar B} [\gamma]
    [\sub{\id_{\interp{(\Delta,\cvar y : C,\cvar f : \arr[C]{\cvar x}{\cvar
    y})}_{\hvar B}}, \hvar{x}\mapsto \hvar{x'},\hvar{y}\mapsto\hvar{y'},\hvar{f}\mapsto\hvar{f'}}] \\
  p
  &\defeq \elim_{\hvar B}(\Delta,A,\Gamma,\gamma\circ
    \sub{\id_{\interp{\Delta}_{\hvar B}}, \hvar y\mapsto \hvar x,
                  \hvar f \mapsto \refl_{\interp{C}_{\hvar B},\hvar x}}).
\end{align*}
Intuitively, in the expression \(\elim(\Delta,A,\Gamma,\gamma)\), the context
\(\Gamma\) is meant to represent the ambiant ps-context we started from,
\(\Delta\) is meant to represent a sub-context of \(\Gamma\) obtained by
successively peeling off variables from \(\Gamma\), and \(\gamma\) a
substitution mapping variables of \(\interp{\Gamma}_{\hvar B}\) to variables of
\(\interp{\Delta}_{\hvar B}\) or reflexivity witnesses of such.

In this mutually inductive definition, all the inductive are strictly
decreasing, except in the case for \(\elim\) where in the inductive the
structural height of the type argument \(A\) is constant, but the first argument
is strictly decreasing. The induction is thus well founded, and proceeds by
structural induction on the syntax, and upon hitting a term constructor
\(\coh_{\Gamma,A}\) temporarily switches to an induction on the length of
\(\Gamma\) before resuming the structural induction. This inductive scheme is
also used for the proof of correctness of
Proposition~\ref{prop:translation-catt}. Note that in the inductive case
defining \(\elim\), the type \(C\) is the type of a ps-context, and thus a type
in \gsett, so the expression \(\interp{C}_{\hvar B}\) denotes the previously
defined translation on \gsett and it not an inductive call.

\begin{restatable}{lemma}{translationfunctor}
\label{lemma:translation-functor}
  The proposed translation scheme respects substitution application at the level
  of the pre-syntax. More precisely, the following results hold:
  \begin{itemize}
  \item For all pre-type \(A\) and all pre-substitution \(\gamma\) in \catt, we
    have the following equality of pre-types in \hott:
    \(\interp{A[\gamma]}_{\hvar B} = \interp{A}_{\hvar
      B}\left[\interp{\gamma}_{\hvar B}\right]\).
  \item For all pre-term \(t\) and all pre-substitution \(\gamma\) in \catt, we
    have the following equality of pre-terms in \hott:
    \(\interp{t[\gamma]}_{\hvar B} = \interp{t}_{\hvar
      B}\left[\interp{\gamma}_{\hvar B}\right]\).
  \item For all pairs of pre-substitutions \(\delta,\gamma\) in \catt, we
    have the following equality of pre-substitutions in \hott:
    \(\interp{\delta\circ\gamma}_{\hvar B} = \interp{\delta}_{\hvar
      B} \circ \interp{\gamma}_{\hvar B}\).
  \end{itemize}
\end{restatable}
\begin{proof}
  One can check this property by structural induction on the pre-syntax of
  \catt, see Appendix~\ref{app:catt-to-hott}.
\end{proof}

\begin{restatable}{lemma}{translationinit}
  \label{lemma:translation-init} For every term
  \(\termrule{(\cvar x:\obj)}{u}{A}\) in \catt, we have
  \(\interp{u}_{\hvar B} \equiv \refl_{\hvar B, \hvar x}^{\dim A + 1}\).
  Similarly, for every type \(\typerule{(\cvar x : \obj)}{A}\) in \catt, we have
  \(\interp{A}_{\hvar B} \equiv\Id^{\dim A+1}_{\hvar B, \hvar x}\)
\end{restatable}
\begin{proof}
  The result on term can be proven by induction together with a corresponding
  for the auxiliary operation \(\elim\). The case of \(\elim\) is itself an
  induction on the length of the ps-context given by the first argument, and
  amounts to successive applications of the computation rule~\eqref{eq:eta-j}.
  The result on types is a consequence of that of terms. See
  Appendix~\ref{app:catt-to-hott} for more details.
\end{proof}

\begin{restatable}{proposition}{translationcatt}
  \label{prop:translation-catt}
  The following results hold:
  \begin{itemize}
  \item For any context \(\ctxrule{\Gamma}\) in \catt, we have
    \(\ctxrule{\interp{\Gamma}_{\hvar{B}}}\) in \hott.
  \item For any type \(\typerule{\Gamma}{A}\) in \catt, we have
    \(\typerule{\interp{\Gamma}_{\hvar{B}}}{\interp{A}_{\hvar B}}\) in \hott.
  \item For any term \(\termrule{\Gamma}{t}{A}\) in \catt, we have
    \(\termrule{\interp{\Gamma}_{\hvar{B}}}{\interp{t}_{\hvar{B}}}{\interp{A}_{\hvar{B}}}\)
    in \hott.
  \item For any substitution \(\subrule{\Delta}{\gamma}{\Gamma}\) in \catt, we
    have
    \(\subrule{\interp{\Delta}_{\hvar{B}}}{\interp{\gamma}_{\hvar
        B}}{\interp{\Gamma}_{\hvar B}}\) in \hott.
  \item For any coherence \(\vdash \coh_{\Gamma,A}\), we have
    \(\termrule{\interp{\Gamma}_{\hvar B}}{\interp{\coh_{\Gamma,A}}_{\hvar
        B}}{\interp{A}_{\hvar B}}\) in \hott.
  \item For any pair of ps-contexts \(\Gamma,\Delta\) together with a type
    \(\typerule{\Gamma}{A}\) in \catt and a substitution
    \(\subrule{\interp{\Delta}_{\hvar B}}{\gamma}{\interp{\Gamma}_{\hvar B}}\)
    in \hott, the following judgement is derivable in \hott:
    \[
      \termrule{\interp{\Delta}_{\hvar B}}{\elim_{\hvar
          B}(\Delta,A,\Gamma,\gamma)}{\interp{A}_{\hvar B}[\gamma]}
    \]
  \end{itemize}
\end{restatable}
\begin{proof}
  One can prove this result by mutual induction on the derivation and induction
  on the length of the ps-context for the last case, following the induction
  scheme defining the translation. The case of \(\elim\) is the interesting one,
  where the case of a ps-context with a single object comes from
  Lemma~\ref{lemma:translation-init}, and for the inductive case, we need to
  check that the application of the \(\J\) rule is valid and has the desired
  type. This is done by induction, and manipulation of the judgements with usual
  lemmas about the structure of the type theory, such as cut admissibility and
  the functoriality of the action of substitutions. A complete proof is given in
  Appendix~\ref{app:catt-to-hott}.
\end{proof}

We now proceed take the liftings of the translation of types and terms in order
to obtain a translation from \catt terms to closed terms in \hott. Given a type
\(\typerule{\Gamma}{A}\) in \catt, we denote \(\interp{\typerule{\Gamma}{A}}\)
the \(\Pi\)-lifting of the type
\(\termrule{\interp{\Gamma}_{\hvar B}}{\interp{A}_{\hvar B}}\) in \hott, and
given a term \(\termrule{\Gamma}{t}{A}\) in \catt, we denote
\(\interp{\termrule{\Gamma}{t}{A}}\) the \(\lambda\)-lifting of the term
\(\termrule{\interp{\Gamma}_{\hvar B}}{\interp{t}_{\hvar B}}{\interp{A}_{\hvar
    B}}\) in \hott. This is independent from the choice of the variable
\(\hvar B\).
\begin{theorem}
  Given a derivable term \(\termrule{\Gamma}{t}{A}\) in \catt, the closed
  \(\lambda\)-term \(\interp{\termrule{\Gamma}{t}{A}}\) is valid in \hott. More
  precisely, the following judgement is derivable in \hott:
  \[
    \termrule{\emptycontext}{\interp{\termrule{\Gamma}{t}{A}}}
    {\interp{\typerule{\Gamma}{A}}}.
  \]

\end{theorem}
\begin{proof}
  This is consequence of Lemma~\ref{lemma:lambda-lifting} and
  Proposition~\ref{prop:translation-catt}
\end{proof}

\section{Leveraging automation}
\label{sec:implem}
We have now finished the presentation of our translation scheme, and we conclude
with a discussion around the implementation of this translation scheme, and
briefly present mechanisation principles available in \catt, to finish with an
example where we generate a proof term in \hott leveraging this mechanisation,
and compare it with a definition of a similar term directly in \hott.

\subsection{Implementation}
We have implemented the translation algorithm presented in
Section~\ref{sec:catt-to-hott}, in the form of a plugin for the \coq proof
assistant, named
\cattplugin\footnote{\url{https://www.github.com/thibautbenjamin/catt}}. This
plugin extends the syntax of \coq with a new primitive \verb|Catt|, that can be
used with the following syntax:
\begin{verbatim}
Catt [names] From File [path]
\end{verbatim}
where \verb|[path]| is the path of a \catt file and \verb|[names]| is a list of
space separated string corresponding to names of declared coherences and terms
in the \catt file. Upon execution, this command uses as a backend \catt to parse
and type-check the provided file, and then uses the translation algorithm to
generate the terms in \coq corresponding to the specified coherences and terms.
The choice of \coq was motivated by the extensibility of its plugin system, and
the ease of tooling, since both \coq and our implementation of \catt are written
in \ocaml. In order to generate terms, our implementation interfaces directly
with the \coq kernel. Following \coq's internals, our implemented algorithm uses
pattern-matching of identities against the reflexivity terms instead of
application of the \(\J\) rule.

Our implementation is still experimental, and could be improved in many ways.
First, as the time of writing, we always expand all applied terms down to
coherences. In practice, this means every function application is inline, making
the type checking by \coq computationally heavy.

\subsection{Mechanisation principles in CaTT}
Despite \hott terms interpret \catt terms, there are still three main
differences between the work in \catt and the work in \hott. First, \catt is a
theory for weak \(\omega\)-categories, and thus is inherently directed, while in
\hott the identity types satisfy symmetry. Second, \catt is fully weak, that is
terms in \catt never reduce to other terms, while in \hott, it is always the
case that the composition of two reflexivity definitionnally reduces to a
reflexivity. Finally, in \hott, all the operations that we define on identity
types are polymorphic, allowing to lift them to identity types on identity types
themselves, and so on, while in \catt, the type \(\obj\) can only ever be
substituted for itself. This may indicate that working in \catt is harder than
working in \hott. However, this added complexity is largely offset by the
various mechanisation features that implemented in \catt, which are only
permitted by the simplicity of the language. The suspension meta-operation in
\catt~\cite{benjaminHomCategoriesComputad2024} exactly accounts for the lack of
polymorphism. The automatic computation of inverses and cancellation
witnesses~\cite{benjamin_invertible_2024} accounts for the lack of symmetry when
the cells still happen to be invertible. Other meta-operations, such as the
opposites~\cite{benjaminHomCategoriesComputad2024} and the
functorialisation~\cite{benja2020} account for a wide range of
phenomena that are not accounted for in \hott. These meta-operations are
leveraged in practical uses of \catt to construct complex terms with a minimal
proof effort.

\subsection{Quantitative experiment: the Eckmann-Hilton cell}
As an example to illustrate our point, we present here the Eckmann-Hilton cell.
This is a cell that formalises the Eckmann-Hilton argument, playing an important
role in topology. It can be derived in \hott by considering a term \(x\) of any
type \(A\) and two terms \(a,b\) of type \(\Id(\refl_{A,x},\refl_{A,x})\), and
is a witness that inhabit the identity type between the transitivity of \(a\)
and \(b\), and the transitivity of \(b\) and \(a\). Deriving this cell is done
by purely algebraic manipulation of identity types, and can be done in \catt.
Figure~\ref{fig:eh} illustrates the intuitive idea allowing to define this cell,
where the one dimensional cells are all identities. This figure is only for
intuition, as it is missing a lot of data, such as the associators and unitors.

\begin{figure}
  \centering
\[\begin{tikzcd}[column sep = small]
  & x
  && x
  &&&& x
  && x \\
  \\
  \\
  x
  && x
  && x
  && x
    && x
  && x
     \arrow[""{name=0, anchor=center, inner sep=0}, "\shortmid"{marking}, curve={height=-30pt}, from=1-2, to=1-4]
     \arrow[""{name=1, anchor=center, inner sep=0}, "\shortmid"{marking}, curve={height=30pt}, from=1-2, to=1-4]
     \arrow[""{name=2, anchor=center, inner sep=0}, "\shortmid"{marking}, from=1-2, to=1-4]
     \arrow[""{name=3, anchor=center, inner sep=0}, "\shortmid"{marking}, curve={height=-30pt}, from=1-8, to=1-10]
     \arrow[""{name=4, anchor=center, inner sep=0}, "\shortmid"{marking}, curve={height=30pt}, from=1-8, to=1-10]
     \arrow[""{name=5, anchor=center, inner sep=0}, "\shortmid"{marking}, from=1-8, to=1-10]
     \arrow[""{name=6, anchor=center, inner sep=0}, "\shortmid"{marking}, from=4-1, to=4-3]
     \arrow[""{name=7, anchor=center, inner sep=0}, "\shortmid"{marking}, curve={height=-30pt}, from=4-1, to=4-3]
     \arrow[""{name=8, anchor=center, inner sep=0}, "\shortmid"{marking}, curve={height=30pt}, from=4-3, to=4-5]
     \arrow[""{name=9, anchor=center, inner sep=0}, "\shortmid"{marking}, from=4-3, to=4-5]
     \arrow[shorten <=10pt, shorten >=10pt, Rightarrow, scaling nfold=3, from=4-5, to=4-7]
     \arrow[""{name=10, anchor=center, inner sep=0}, "\shortmid"{marking}, curve={height=30pt}, from=4-7, to=4-9]
     \arrow[""{name=11, anchor=center, inner sep=0}, "\shortmid"{marking}, from=4-7, to=4-9]
     \arrow[""{name=12, anchor=center, inner sep=0}, "\shortmid"{marking}, curve={height=-30pt}, from=4-9, to=4-11]
     \arrow[""{name=13, anchor=center, inner sep=0}, "\shortmid"{marking}, from=4-9, to=4-11]
     \arrow["a", shorten <=4pt, shorten >=4pt, Rightarrow, from=0, to=2]
     \arrow["b", shorten <=4pt, shorten >=4pt, Rightarrow, from=2, to=1]
     \arrow[shorten <=8pt, shorten >=12pt, Rightarrow, scaling nfold=3, from=1, to=4-3]
     \arrow["b"', shorten <=4pt, shorten >=4pt, Rightarrow, from=3, to=5]
     \arrow["a"', shorten <=4pt, shorten >=4pt, Rightarrow, from=5, to=4]
     \arrow["a"', shorten <=4pt, shorten >=4pt, Rightarrow, from=7, to=6]
     \arrow["b"', shorten <=4pt, shorten >=4pt, Rightarrow, from=9, to=8]
     \arrow["a"', shorten <=4pt, shorten >=4pt, Rightarrow, from=11, to=10]
     \arrow[shorten <=12pt, shorten >=8pt, Rightarrow, scaling nfold=3, from=4-9, to=4]
     \arrow["b"', shorten <=4pt, shorten >=4pt, Rightarrow, from=12, to=13]
\end{tikzcd}\]
\caption{Intuitive illustration of the Eckmann-Hilton cell}
\label{fig:catt-eh}
\end{figure}
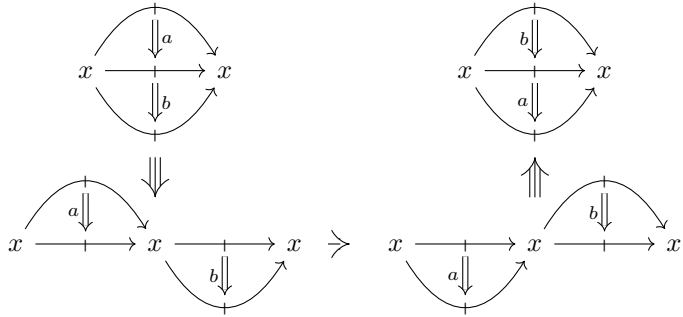

We have a \catt definition of the Eckmann-Hilton cell using all the
mechanisation features mentioned above (see Appendix~\ref{app:eh-catt}), using
786 characters. The definition of the Eckmann-Hilton cell in the \hott library
of \coq\footnote{\url{https://github.com/HoTT/Coq-HoTT}} is in the file
\verb|theories/Homotopy/Syllepsis.v|. A direct comparison with this file would
not be fair, since it relies on modules that provide a lot more than just the
required lemmas for defining the Eckmann-Hilton cell. Nevertheless, a
qualitative comparison shows that the \catt definition is far shorter. In order
to assess our translation scheme, we compare a definition of the Eckmann-Hilton
cell obtained by generation from \catt with the definition that is in the \hott
library. As a proxy for the complexity of a term, we count the number of
characters that \coq prints when asked to print the term in normal form with
\verb|Eval cbv in [term]|. When running the experiment, we found the complexity
of the term generated by \catt to be of 79773 characters long, while that of the
term from \hott is 428183 characters long. These experiments were conducted in a
slightly modified version of our
implementation\footnote{\url{https://github.com/thibautbenjamin/catt/tree/catt-vs-hott/coq_plugin/catt-vs-hott}}
where the generated names of variables have been shortened to use the same name
as in \catt, as to not artificially skew the results, since our actual
implementation generates very long variable names. We also used the
\verb|Set Printing All| option, to avoid hidden confusion factors, such as
conveniently defined notations and implicit arguments. Overall, we deem these
results surprising, since we expect \catt to generally generate longer terms,
due to the fact that it is inherently directed and fully weak, although not very
significant due to the proxy for measuring complexity that we chose. We believe
that this result is mostly due to the definition of the Eckmann-Hilton cell in
\catt being more parsimonious, and we take this as a witness that \catt is a
helpful tool when reasoning on higher categories making such a definition
achievable, but mostly we take these results as witnesses that the complexity of
the generated is reasonably similar to that of similar terms that one would
define by hand.

\bibliographystyle{plain}
\bibliography{biblio}
\newpage

\appendix
\section{Structural rules of dependent type theories}
\label{app:cwf}
We consider dependent type theories that always satisfy some amount of
structure, namely that which makes their syntax into a category with families.
We recall here part of this structure that we use, and which is satisfied by all
the types theories we consider. First, the variables are pre-terms, and we want
there to be an action of substitutions on types and terms. On term type and term
constructors, this action is defined for every type and term constructor, but is
usually the unique sensible choice. On variables, this action associates the
pre-term associated to the variable in the substitution, or leaves the variable
unchanged if there is no association. We can express this inductively as follows
\begin{align*}
  x[\sub{}] & =  \textsf{undefined}
  & x[\sub{\gamma,y\mapsto t}]
  &=
    \begin{cases}
      t & \text{if \(y=x\)}\\
      x[\gamma] & \text{otherwise}
    \end{cases}
\end{align*}
There is also a composition of substitutions defined inductively from the action
on terms as follows
\begin{align*}
  \sub{}\circ\gamma & = \sub{}
  & \sub{\delta,x\mapsto t}\circ\gamma
  &= \sub{\delta\circ\gamma,x\mapsto t[\gamma]}.
\end{align*}
All the substitutions that we consider here are implicitly well-scoped, and so
in particular, when we write \(t[\gamma]\), (reps.~\(A[\gamma]\),
\(\delta\circ\gamma\)), we implicitly assume that all free variables of \(t\),
(resp.~\(A\),\(\delta\)) are bound in \(\gamma\). Given a pre-context
\(\Gamma\), we denote \(\id_{\Gamma}\) the identity substitution on \(\Gamma\)
which maps every variable to itself. It is required that the action of
substitutions on pre-types and on pre-terms is compatible with the substitution,
and that the identity is a neutral element for this action in the sense that the
following equations hold:
\begin{align}
  \label{eq:functorial-action}
  \begin{aligned}
    A[\delta\circ\gamma] &= A[\delta][\gamma]
    &
    t[\delta\circ\gamma] &= t[\delta][\gamma] \\
    A[\id_{\Gamma}] &= A & t[\id_{\Gamma}] &= t
    & \gamma\circ\id_{\Gamma} &= \gamma.
  \end{aligned}
\end{align}
The derivation rules for contexts, substitutions and terms that are variables
are standard, and given by the following:
  \begin{mathpar}
    \inferdef{(\(\emptycontext\)-ctx)}{\null}{\ctxrule{\emptycontext}}
    \label{rule:ce}
    \and
    \inferdef{(+-ctx)}{\ctxrule{\Gamma} \\ \typerule{\Gamma}{A}}
    {\ctxrule{(\Gamma,x:A)}}\label{rule:cc} \\
    \inferdef{(\(\emptycontext\)-sub)}{\Gamma}{\subrule{\Gamma}
      {\sub{}}{\emptycontext}}
    \label{rule:se}
    \and
    \inferdef{(+-sub)}{
      \subrule{\Delta}{\gamma}{\Gamma} \\
      \typerule{\Gamma}{A}\\
      \termrule{\Delta}{t}{A[\gamma]}}
    {\subrule{\Delta}{\sub{\gamma,x\mapsto t}}{(\Gamma,x:A)}}
    \label{rule:sc}\\
    \inferdef{(var-in)}{\ctxrule{\Gamma} \\ \assoc(\Gamma,x) = A}
    {\termrule{\Gamma}{x}{A}}
    \label{rule:var}
  \end{mathpar}
where \(\assoc(\Gamma,x)\) is the partial function that returns the last element
to which \(x\) is associated if it exists, in the context \(\Gamma\) seen as an
association list. We also require our type theories to satisfy cut-admissibility
for terms and types, that is we require the following rules to be admissible
\begin{align}
  \label{eq:cut-admissibility}
  \inferrule{\typerule{\Gamma}{A} \quad
  \subrule{\Delta}{\gamma}{\Gamma}}
  {\typerule{\Delta}{A[\gamma]}}
  &&
     \inferrule{\termrule{\Gamma}{t}{A} \quad
  \subrule{\Delta}{\gamma}{\Gamma}}
  {\termrule{\Delta}{t[\gamma]}{A[\gamma]}}
  &&
     \inferrule{\subrule{\Gamma}{\vartheta}{\Theta} \quad
  \subrule{\Delta}{\gamma}{\Gamma}}
  {\termrule{\Delta}{\vartheta\circ\gamma}{\Theta}}
\end{align}

It is standard that Martin-Löf type theory, homotopy type theory or any
reasonable fragment of those respect all of the structure described above. The
dependent type theories \gsett and \catt also do, as witnessed by a
formalisation in \agda~\cite{benjamin_formalization_2021}. In the case of the
\catt, the action of a substitution \(\delta\) on a pre-term of the form
\(\coh_{\Gamma,A}[\gamma]\) is given by the formula
\[
  \coh_{\Gamma,A}[\gamma][\delta] = \coh_{\Gamma,A}[\gamma\circ\delta].
\]
The dependent type theory \hott is also subject to a non-trivial definitional
equality, for which we assume that we have a confluent rewriting system, and we
assume that this equality respects the typing, in the sense that the following
rules are admissible
\begin{align}
  \label{eq:congruence}
  \inferrule{\typerule{\Gamma}{A}\\ A\equiv B}{\typerule{\Gamma}{B}}
  &&
     \inferrule{\termrule{\Gamma}{u}{A}\\ A\equiv B}{\termrule{\Gamma}{u}{B}}
  &&
     \inferrule{\termrule{\Gamma}{u}{A}\\ u\equiv v}{\termrule{\Gamma}{v}{A}}
\end{align}

\section{Proof of Section~\ref{sec:gsett-to-hott}}
\label{app:gsett-to-hott}
This section is dedicated to the proof of Lemma~\ref{lemma:translation-gsett},
which shows the correctness of the translation principle from \gsett to \hott.

\translationgsett*
\begin{proof}
  We proceed by mutual induction on the derivation tree of the judgements.
  \begin{itemize}
  \item If the derivation tree ends with~\ruleref{rule:ce}, then it is a
    derivation of the judgement \(\ctxrule{\emptycontext}\). We then have
    \(\interp{\emptycontext}_{\hvar B} = (\hvar B : \Type)\), and the judgement
    \(\ctxrule{(\hvar B : \Type)}\) is derivable in \hott.
  \item If the derivation tree ends with~\ruleref{rule:cc}, then it is a
    derivation of the judgement \((\Gamma,\cvar {x}:A)\), produced from
    derivations for the two following judgements in \gsett:
    \begin{align*}
      \ctxrule{\Gamma} & & \typerule{\Gamma}{A}
    \end{align*}
    by induction, these give derivations for the following judgements in \hott
    \begin{align*}
      \ctxrule{\interp{\Gamma}_{\hvar B}} & & \typerule{\interp{\Gamma}_{\hvar
                                              B}}{\interp{A}_{\hvar B}}
    \end{align*}
    which by application of~\ruleref{rule:cc} give a derivation of the judgement
    \(\ctxrule{(\interp{\Gamma}_{\hvar B},\hvar x : \interp{A}_{\hvar B})}\) in
    \hott.
  \item A derivation ending with~\ruleref{rule:obj} is a derivation of
    \(\typerule{\Gamma}{\obj}\) obtained from a derivation of
    \(\ctxrule{\Gamma}\). By induction, this gives a derivation of the judgement
    \(\ctxrule{\interp{\Gamma}_{\hvar B}}\). Moreover, this contexts starts with
    the association \((\hvar {B} : \Type)\), and since the variable \(\hvar B\)
    is not a valid variable name in \gsett, this association cannot be
    overridden in \(\interp{\Gamma}_{\hvar B}\). Thus,~\ruleref{rule:var} gives
    a derivation of the judgement \(\termrule{\Gamma}{\hvar B}{\Type}\) in
    \hott, which in terms let us get by~\ruleref{rule:el} a derivation of
    \(\typerule{\Gamma}{\El(\hvar B)}\).
  \item A derivation ending with~\ruleref{rule:arr} is a derivation of
    \(\typerule{\Gamma}{\arr[A]{\cvar x}{\cvar y}}\) obtained from derivations
    of the following three judgements in \gsett:
     \begin{align*}
      \typerule{\Gamma}{A} & & \termrule{\Gamma}{\cvar x}{A}
      & & \termrule{\Gamma}{\cvar y}{A}
    \end{align*}
    By induction, each of them gives a derivation for the corresponding
    judgement in \hott
    \begin{align*}
      \typerule{\interp{\Gamma}_{\hvar B}}{\interp{A}_{\hvar B}}
      & & \termrule{\interp{\Gamma}_{\hvar B}}{\hvar x}{\interp{A}_{\hvar B}}
      & & \termrule{\interp{\Gamma}_{\hvar B}}{\hvar y}{\interp{A}_{\hvar B}}
    \end{align*}
    \ruleref{rule:id} then gives a derivation of the judgement
    \(\typerule{\interp{\Gamma}_{\hvar B}}{\Id_{\interp{A}_{\hvar B}}(\hvar x,
      \hvar y)}\) in \hott.
  \item A derivation finishing with an application of~\ruleref{rule:var} is
    a derivation of the judgement \(\termrule{\Gamma}{\cvar x}{A}\) obtained
    from a derivation of \(\ctxrule{\Gamma}\), such that
    \(\assoc(\Gamma,\cvar x)=A\). By induction, we get a derivation for the rule
    \(\ctxrule{\interp{\Gamma}_{\hvar B}}\) in \hott. Moreover, from the
    definition of \(\interp{\Gamma}_{\hvar B}\), one can see that
    \(\assoc(\interp{\Gamma}_{\hvar B},\hvar x) = \interp{A}_{\hvar B}\). Hence,
    we get a definition of the judgement
    \(\termrule{\interp{\Gamma}_{\hvar B}}{\hvar x}{\interp{A}_{\hvar B}}\).
  \end{itemize}
\end{proof}

\section{Proof of Section~\ref{sec:catt-to-hott}}
\label{app:catt-to-hott}

This section is dedicated to proving correctness of the translation scheme of
\catt term to closed \(\lambda\)-terms in \hott.

\translationfunctor*
\begin{proof}
  These results hold syntactically regardless of any well-foundedness principle.
  We fix a pre-substitution \(\gamma\), and proceed by mutual induction.
  \begin{itemize}
  \item Considering the pre-type \(\obj\), we have
    \(\interp{\obj[\gamma]}_{\hvar B} = {\hvar B}\). On the other hand,
    \({\interp{\gamma}}_{\hvar B}\) starts by defining the mapping
    \(\hvar B \mapsto \hvar B\), and since by assumption, the variable
    \(\hvar B\) is not a variable in \(\catt\), this mapping is never overridden
    in \(\interp{\gamma}_{\hvar B}\). Thus we have
    \(\hvar B[\interp{\gamma}_{\hvar B}] = \hvar B\). This shows the equality
    \[
      \interp{\obj[\gamma]}_{\hvar B} = (\interp{\obj}_{\hvar
        B})[\interp{\gamma}_{\hvar B}].
    \]
  \item Considering the pre-type \(\arr[A]{u}{v}\), by the inductive case for
    types, we have the equality
    \(\interp{A[\gamma]}_{\hvar B} = (\interp{A}_{\hvar
      B})[\interp{\gamma}_{\hvar B}]\). By the mutually inductive case for
    terms, we have
    \(\interp{u[\gamma]}_{\hvar B} = (\interp{u}_{\hvar
      B})[\interp{\gamma}_{\hvar B}]\), and similarly for \(v\). this shows the
    equality
    \[
      \interp{(\arr[A]{u}{v})[\gamma]}_{\hvar B} =
      (\interp{\arr[A]{u}{v}}_{\hvar B})[\interp{\gamma}_{\hvar B}].
    \]
  \item Considering a variable \(\cvar x\) in \catt, denote
    \(u = \cvar x[\gamma]\) the term corresponding to the last binding of
    \(\cvar x\) in \(\gamma\). By definition, this binding defines the binding
    \(\hvar x \mapsto \interp{u}_{\hvar B}\) in \(\interp{\gamma}_{\hvar B}\).
    Moreover, the binding \(\cvar x \mapsto u\) being the last binding of
    \(\cvar x\) in \(\gamma\) implies that this binding of \(\hvar x\) is never
    overridden in \(\interp{\gamma}_{\hvar B}\). Hence, we have
    \(\hvar x[\interp{\gamma}_{\hvar B}] = \interp{u}_{\hvar B}\). This proves
    the equality
    \[
      \interp{\cvar x[\gamma]}_{\hvar B} = (\interp{\cvar x}_{\hvar
        B})[\interp{\gamma}_{\hvar B}]
    \]
  \item Considering the pre-term \(\coh_{\Gamma,A}[\delta]\), we have by
    definition
    \begin{align*}
      \interp{\coh_{\Gamma,A}[\delta][\gamma]}_{\hvar B}
      & =
        (\interp{\coh_{\Gamma,A}}_{\hvar B})[\interp{\delta\circ \gamma}_{\hvar
        B}] \\
      \interp{\coh_{\Gamma,A}[\delta]}_{\hvar B}[\interp{\gamma}_{\hvar B}]
      & = (\interp{(\coh_{\Gamma,A})}_{\hvar B})[\interp{\delta}_{\hvar
        B}][\interp{\gamma}_{\hvar B}]
    \end{align*}
    By the inductive case for substitutions, we have
    \(\interp{\delta \circ\gamma}_{\hvar B} = \interp{\delta}_{\hvar B}\circ
    \interp{\gamma}_{\hvar B}\). Then using the definition of the action of
    substitution on coherences as well as \eqref{eq:functorial-action} shows
    that
    \[
      \interp{\coh_{\Gamma,A}[\delta][\gamma]}_{\hvar B} =
      \interp{\coh_{\Gamma,A}[\delta]}_{\hvar B}[\interp{\gamma}_{\hvar B}]
    \]

  \item Considering the empty pre-substitution \(\sub{}\), since
    \(\hvar B[\interp{\gamma}_{\hvar B}] = \hvar B\), we have the following
    equalities
    \begin{align*}
      \interp{\sub{}\circ\gamma}_{\hvar B}
      &= \interp{\sub{}}_{\hvar B} = \sub{\hvar B\mapsto \hvar B} \\
      \interp{\sub{}}_{\hvar B} \circ \interp{\gamma}_{\hvar B}
      &= \sub{\hvar B\mapsto \hvar B[\interp{\gamma}_{\hvar B}]} = \sub{\hvar B
        \mapsto \hvar B}.
    \end{align*}
  \item Considering the pre-substitution \(\sub{\delta,\cvar x\mapsto u}\), the
    inductive case shows that
    \(\interp{\delta\circ\gamma}_{\hvar B} = \interp{\delta}_{\hvar B}\circ
    \interp{\gamma}_{\hvar B}\). The mutually inductive case for terms shows
    that
    \(\interp{u[\gamma]}_{\hvar B} = (\interp{u}_{\hvar
      B})[\interp{\gamma}_{\hvar B}]\). Together, they prove the equality
    \[
      \interp{\sub{\delta,\cvar x\mapsto u}\circ \gamma}_{\hvar B} =
      \interp{\sub{\delta,\cvar x\mapsto u}}_{\hvar B}\circ
      \interp{\gamma}_{\hvar B}. \qedhere
    \]
  \end{itemize}
\end{proof}

Before proving Lemma~\ref{lemma:translation-init}, we first show the following
result as an intermediate step
\begin{lemma}
  \label{lemma:vars-in-elim}
  Given two ps-contexts \(\Gamma\) and \(\Delta\), a type
  \(\typerule{\Gamma}{A}\) and a substitution
  \(\subrule{\Delta}{\gamma}{\Gamma}\) in \catt, the free variables of the term
  \(\elim_{\hvar B}(\Delta,A,\Gamma,\gamma)\) are all bound in the context
  \(\interp{\Delta}_{\hvar B}\).
\end{lemma}
\begin{proof}
  We can verify this by induction on the ps-context \(\Delta\).
  \begin{itemize}
  \item For the context \((\cvar x : \obj)\), the free variables of the term
    \[
      \elim_{\hvar B}((\cvar x: \obj),A,\Gamma,\gamma)=\refl_{\hvar B,\hvar x}
    \]
    are \(\hvar B\) and \(\hvar x\) and both are bound in the context
    \(\interp{(\cvar x : \obj)}_{\hvar B} = (\hvar B : \Type, \hvar x :
    \El(\hvar B))\).
  \item For the context
    \((\Delta,\cvar y : C, \cvar f : \arr[C]{\cvar x}{\cvar y})\), recall that
    we have
    \[
      \elim_{\hvar B}((\Delta,\cvar y:C,\cvar{f}:\arr[C]{\cvar x}{\cvar y}),
      A,\Gamma,\gamma) = \J({\interp{C}_{\hvar B}}, \hvar x,
      P(\hvar{x'},\hvar{y'},\hvar{f'}), p)\ \hvar y\ \hvar f
\]
where
\begin{align*}
  P(\hvar{x'},\hvar{y'},\hvar{f'})
  & \defeq \interp{A}_{\hvar B} [\gamma]
    [\sub{\id_{\interp{(\Delta,\cvar y : C,\cvar f : \arr[C]{\cvar x}{\cvar
    y})}_{\hvar B}}, \hvar{x}\mapsto \hvar{x'},\hvar{y}\mapsto\hvar{y'},\hvar{f}\mapsto\hvar{f'}}] \\
  p &\defeq \elim_{\hvar B}(\Delta,A,\Gamma,\gamma\circ
      \sub{\id_{\interp{\Delta}_{\hvar B}}, \hvar y\mapsto \hvar x,
      \hvar f \mapsto \refl_{\interp{C}_{\hvar B},\hvar x}}).
\end{align*}
Hence, the free variables of the term
\(\elim_{\hvar B}((\Delta,\cvar y:C,\cvar{f}:\arr[C]{\cvar x}{\cvar y}),
A,\Gamma,\gamma)\) are the union of the free variables of the term \(p\) and the
set \(\{\hvar x,\hvar y, \hvar f\}\). By induction, all free variables of \(p\)
are bound in \(\interp{\Delta}_{\hvar B}\), thus, all the above variables are
bound in
\(\interp{(\Delta,\cvar y : C , \cvar f : \arr[\interp{C}_{\hvar B}]{\hvar
    x}{\hvar y})}_{\hvar B}\)
  \end{itemize}
\end{proof}

\translationinit*
\begin{proof}
  We first prove the result for terms, that we prove mutually inductively with
  the following proposition: for every ps-contexts \(\Gamma,\Delta\), type
  \(\typerule{\Gamma}{A}\) and substitution
  \(\subrule{\interp{\Delta}_{\hvar B}}{\gamma}{\interp{\Gamma}_{\hvar B}}\), as
  well as a substitution \(\subrule{(\cvar x : \obj)}{\delta}{\Delta}\), we have
  \[
    (\elim_{\hvar B}(\Delta,A,\Gamma,\gamma))[\interp{\delta}_{\hvar B}]\equiv
    \refl^{\dim A+1}_{\hvar B,\hvar x}.
  \]
  We proceed by structural induction on the term, and induction on the length of
  the ps-context.
  \begin{itemize}
  \item For the term \(\cvar {x}\) which is the unique variable of the context,
    by definition, we have
    \(\interp{\cvar x}_{\hvar B} = \hvar x = \refl_{\hvar B, \hvar x}^{0}\).
  \item For a coherence term of the form \(\coh_{\Gamma,C}[\gamma]\) with
    \(C[\gamma] = A\). Then
    \[
      \interp{\coh_{\Gamma,C}[\gamma]}_{\hvar B} = (\elim_{\hvar
        B}(\Gamma,C,\Gamma,\id_{\Gamma}))[\interp{\gamma}_{\hvar B}],
    \]
    so by the mutually inductive proposition, and using that \(\dim A = \dim
    C\), we get
    \begin{align*}
      \interp{\coh_{\Gamma,C}[\gamma]}_{\hvar B}
      &= \refl^{\dim C+1}_{\hvar B, \hvar x}\\
      &= \refl^{\dim A+1}_{\hvar B, \hvar x}.
    \end{align*}
  \item Given a ps-context \(\Gamma\) with a substitution
    \(\subrule{\interp{(\cvar y : \obj)}_{\hvar
        B}}{\gamma}{\interp{\Gamma}_{\hvar B}}\) and a substitution
    \(\subrule{(\cvar x : \obj)}{\delta}{(\cvar y:\obj)}\), the only term of
    type \(\obj\) derivable in the context \((\cvar x : \obj)\) is \(\cvar x\)
    thus we have \(\delta = \sub{\cvar y\mapsto \cvar x}\). Then, we have
    \begin{align*}
      \elim((\cvar y : \obj),A,\Gamma,\gamma)[\interp{\delta}_{\hvar B}]
      &= \refl^{\dim A + 1}_{\hvar B,\hvar y}[\hvar B\mapsto\hvar B,\hvar
        y\mapsto \hvar x]\\
      &= \refl^{\dim A+1}_{\hvar B,\hvar x}.
    \end{align*}
  \item Given a ps-context \(\Gamma\) and a ps-context of the form
    \(\Delta' = (\Delta,\cvar z : D, \cvar f : \arr[D]{\cvar y}{\cvar z})\),
    with a substitution
    \(\subrule{\interp{\Delta'}_{\hvar B}}{\gamma}{\interp{\Gamma}_{\hvar B}}\)
    and a substitution \(\subrule{(\cvar x : \obj)}{\delta'}{\Delta'}\), the
    substitution \(\delta'\) decomposes as
    \(\delta'=\sub{\delta,\cvar z \mapsto v, \cvar f \mapsto w}\).
    We have
    \begin{align*}
      \elim(\Delta',A,\Gamma,\gamma){\interp{\delta'}_{\hvar B}}
      &\equiv
        \J
        ({\interp{C}_{\hvar B}},
        {\hvar y},
        {P(\hvar{y'},\hvar{z'},\hvar{f'})},
        p)[\interp{\delta'}_{\hvar B}]\
        (\interp{v}_{\hvar B})\
        (\interp{w}_{\hvar B})
    \end{align*}
    where
    \begin{align*}
      P(\hvar{y'},\hvar{z'},\hvar{f'})
      &
        \defeq \interp{C}_{\hvar B} [\gamma] [\sub{\id_{\interp{\Delta'}_{\hvar
        B}},\hvar{y}\mapsto \hvar{y'},\hvar{z}\mapsto \hvar{z'},\hvar{f}\mapsto \hvar{f'}}] \\
      p &\defeq \elim_{\hvar B}(\Delta,A,\Gamma,\gamma\circ
          \sub{\id_{\interp{\Delta}_{\hvar B}}, \hvar z\mapsto \hvar y,
          \hvar f \mapsto \refl_{\interp{C}_{\hvar B},\hvar
          y}}).
    \end{align*}
    By the
    structural induction case for terms, we have
    \begin{align*}
      \interp{v}_{\hvar B} &\equiv \refl^{\dim D + 1}_{\hvar B,\hvar x}
      & \interp{w}_{\hvar B} &\equiv \refl^{\dim D + 2}_{\hvar B,\hvar x},
    \end{align*}
    and since \(\dim D + 2 \geq 1\), \(\refl^{\dim D+2}_{\hvar B,\hvar x}\) is a
    reflexivity term the computation rule~\eqref{eq:eta-j} shows
    \begin{align*}
      \elim(\Delta',A,\Gamma,\gamma){\interp{\delta'}_{\hvar B}}
      &\equiv p[\interp{\delta'}_{\hvar B}].
    \end{align*}
    By Lemma~\ref{lemma:vars-in-elim}, we have
    \(p[\interp{\delta'}_{\hvar B}] \equiv p[\interp{\delta}_{\hvar B}]\). By
    \eqref{eq:cut-admissibility} we have
    \[
      \subrule{\interp{\Delta}_{\hvar
        B}}{\gamma\circ\sub{\id_{\interp{\Delta}_{\hvar B}},\hvar y\mapsto \hvar
        x,\hvar f\mapsto \refl_{\interp{C}_{\hvar B},\hvar
          x}}}{\interp{\Gamma}_{\hvar B}},
  \]
  and by inversion, we have \(\subrule{(\cvar x:\obj)}{\delta}{\Delta}\), thus
  by induction, this shows that
  \(p[\interp{\delta}_{\hvar B}] \equiv \refl^{\dim A+1}_{\hvar B,\hvar x}\).
  \end{itemize}

  The result of type is a consequence of the one on terms. Indeed, we proceed by
  structural induction on the type. For the type \(\obj\), since we have
  \(\interp{\obj}_{\hvar B} = \hvar B = \Id^{0}_{\hvar B, \hvar x}\), and
  considering the type \(\typerule{(\cvar x:\obj)}{\arr[A]{u}{v}}\) in \catt, we
  have derivations of the judgements \(\termrule{(\cvar:\obj)}{u}{A}\) and
  \(\termrule{(\cvar x :\obj)}{v}{A}\), and so by induction and by the result on
  terms, we get the following equalities
  \begin{align*}
    \interp{A}_{\hvar B} &\equiv \Id^{\dim A+1}_{\hvar B,\hvar x}
    & \interp{u}_{\hvar B} &\equiv \refl^{\dim A+1}_{\hvar B,\hvar x}
    & \interp{v}_{\hvar B} &\equiv \refl^{\dim A+1}_{\hvar B,\hvar x}
  \end{align*}
  This shows that we have
  \[
    \interp{\arr[A]{u}{v}}_{\hvar B}\equiv \Id_{\Id^{\dim A +1}_{\hvar B,\hvar
        x}}(\refl^{\dim A+1}_{\hvar B,\hvar x},\refl^{\dim A+1}_{\hvar B},\hvar
    x) = \Id^{\dim A+2}_{\hvar B,\hvar x}. \qedhere
  \]
\end{proof}

\translationcatt*
\begin{proof}
  We prove this result by mutual induction, the induction scheme following that
  of the definition of the translation function \(\interp{\_}_{\hvar B}\) and of
  \(\elim_{\hvar B}\). The cases of contexts, types and variables are similar to
  those of the proof of Lemma~\ref{lemma:translation-gsett}.
  \begin{itemize}
  \item For the empty context \(\ctxrule{\emptycontext}\) in \catt, obtained by
    \ruleref{rule:ce}, we have
    \(\interp{\emptycontext}_{\hvar B} = (\hvar B : \Type)\), and the rules of
    \hott allow to build a derivation of \(\ctxrule{(\hvar B : \Type)}\).
  \item For the context \(\ctxrule{(\Gamma,x:A)}\) in \catt obtained by
    application of \ruleref{rule:cc} from derivations of \(\ctxrule{\Gamma}\)
    and \(\typerule{\Gamma}{A}\), by induction we have derivations for the two
    following judgements in \hott
    \begin{align*}
      \ctxrule{\interp{\Gamma}_{\hvar B}}
      & & \typerule{\interp{\Gamma}_{\hvar B}}{\interp{A}_{\hvar B}}
    \end{align*}
  \item For the type \(\typerule{\Gamma}{\obj}\) in \catt obtained by
    \ruleref{rule:obj} from a derivation of \(\ctxrule{\Gamma}\), by induction
    we have a derivation of \(\ctxrule{\interp{\Gamma}_{\hvar B}}\) in \hott.
    Moreover, the first assignation in this context is \((\hvar B : \Type)\) and
    since the name \(\hvar B\) is not a valid variable in \catt, it cannot be
    overridden in \(\interp{\Gamma}_{\hvar B}\), thus showing that
    \(\assoc(\interp{\Gamma}_{\hvar B}) = \Type\). This gives by applying
    \ruleref{rule:var} and then \ruleref{rule:el}, a derivation of
    \[
      \typerule{\Gamma}{\El(\hvar B)}.
    \]
  \item For the type \(\typerule{\Gamma}{\arr[A]{u}{v}}\) in \catt, obtained by
    \ruleref{rule:arr} from derivations for the following three judgements
    \begin{align*}
      \typerule{\Gamma}{A} && \termrule{\Gamma}{u}{A} && \termrule{\Gamma}{v}{A}
    \end{align*}
    we obtain by the induction cases for types and terms, a derivation for the
    three corresponding judgements in \hott
    \begin{align*}
      \typerule{\interp{\Gamma}_{\hvar B}}{\interp{A}_{\hvar B}}
      && \termrule{\interp{\Gamma}_{\hvar B}}{\interp{u}_{\hvar B}}{\interp{A}_{\hvar B}}
      && \termrule{\interp{\Gamma}_{\hvar B}}{\interp{v}_{\hvar
         B}}{\interp{A}_{\hvar B}}
    \end{align*}
    Then we obtain by \ruleref{rule:id}, a derivation of
    \(\typerule{\interp{\Gamma}_{\hvar B}}{\Id_{\interp{A}_{\hvar
          B}}(\interp{u}_{\hvar B},\interp{v}_{\hvar B})}\).
  \item For a term of \catt given as a variable
    \(\termrule{\Gamma}{\cvar x}{A}\) obtained by \ruleref{rule:var} from a
    derivation of \(\ctxrule{\Gamma}\) such that \(\assoc(\Gamma,\cvar x) = A\),
    we obtain by induction a derivation of
    \(\ctxrule{\interp{\Gamma}_{\hvar B}}\). Moreover, by definition of
    \(\interp{\Gamma}_{\hvar B}\), we have
    \(\assoc(\interp{\Gamma}_{\hvar B},\hvar x) = \interp{A}_{\hvar B}\), which
    lets us derive by \ruleref{rule:var} a derivation of the following judgement
    in \hott
    \[
      \termrule{\Gamma}{\hvar x}{\interp{A}_{\hvar B}}.
    \]
  \item For a term of the form
    \(\termrule{\Delta}{\coh_{\Gamma,A}[\gamma]}{A[\gamma]}\) in \catt obtained
    by application of \ruleref{rule:coh} from derivations of
    \(\cohrule{\Gamma}{A}\) and \(\subrule{\Delta}{\gamma}{\Gamma}\), we get by
    the induction cases for coherences and substitutions derivations for the
    corresponding judgements in \hott:
    \begin{align*}
      \termrule{\interp{\Gamma}_{\hvar B}}{\interp{\coh_{\Gamma,A}}_{\hvar
      B}}{\interp{A}_{\hvar B}}
      && \subrule{\interp{\Delta}_{\hvar B}}{\interp{\gamma}_{\hvar
         B}}{\interp{\Gamma}_{\hvar B}}
    \end{align*}
    and we then obtain, by~\eqref{eq:cut-admissibility},
    Lemma~\ref{lemma:translation-functor} and~\eqref{eq:congruence} a
    derivation of the judgement
    \[
      \termrule{\interp{\Delta}_{\hvar
          B}}{\interp{\coh_{\Gamma,A}[\gamma]}_{\hvar
          B}}{\interp{A[\gamma]}_{\hvar B}}.
    \]
  \item For the empty substitution \(\subrule{\Gamma}{\sub{}}{\emptycontext}\)
    in \catt obtained by \ruleref{rule:se} from a derivation of
    \(\ctxrule{\Gamma}\), we have by induction a derivation of
    \(\ctxrule{\interp{\Gamma}_{\hvar B}}\). Moreover, as proven in the case for
    the type \(\typerule{\Gamma}{\obj}\), we can construct a derivation of
    \(\termrule{\interp{\Gamma}_{\hvar B}}{\hvar B}{\Type}\), which lets us use
    \ruleref{rule:se} and \ruleref{rule:sc}to produce a derivation of
    \[
      \subrule{\interp{\Gamma}_{\hvar B}}{\sub{\hvar B \mapsto \hvar B}}{(\hvar
        B : \Type)}
    \]
  \item For a substitution of the form
    \(\subrule{\Delta}{\sub{\gamma,x\mapsto t}}{(\Gamma,x:A)}\) in \catt
    obtained by \ruleref{rule:sc} from derivations of the following judgements
    \begin{align*}
      \subrule{\Delta}{\gamma}{\Gamma} && \typerule{\Gamma}{A} && \termrule{\Delta}{t}{A[\gamma]}
    \end{align*}
    we get by the induction cases for substitutions, contexts and terms, a
    derivation of the corresponding judgements in \hott
    \begin{align*}
      \subrule{\interp{\Delta}_{\hvar B}}{\interp{\gamma}_{\hvar
      B}}{\interp{\Gamma}_{\hvar B}}
      && \typerule{\interp{\Gamma}_{\hvar B}}{\interp{A}_{\hvar B}}
      && \termrule{\interp{\Delta}_{\hvar B}}{\interp{t}_{\hvar
         B}}{\interp{A[\gamma]}_{\hvar B}}
    \end{align*}
    Applying Lemma~\ref{lemma:translation-functor} then allows us to apply
    \ruleref{rule:sc} to obtain a derivation of the judgement
    \[
      \subrule{\interp{\Delta}_{\hvar B}}{\sub{\interp{\gamma}_{\hvar B}, \hvar
          x \mapsto \interp{t}_{\hvar B}}}{(\interp{\Gamma}_{\hvar B},\hvar x :
        \interp{A}_{\hvar B})}
    \]

  \item For a coherence \(\cohrule{\Gamma}{A}\) in \catt obtained by
    \ruleref{rule:coh-wd}, we have derivations of \(\psrule{\Gamma}\) and
    \(\typerule{\Gamma}{A}\), thus, we get by the inductive case for
    \(\elim_{\hvar B}\) as well as~\eqref{eq:functorial-action} and the fact
    that \(\interp{\id_{\Gamma}}_{\hvar B} = \id_{\interp{\Gamma}_{\hvar B}}\) a
    derivation of
    \[
      \termrule{\interp{\Gamma}_{\hvar B}}{\interp{\coh_{\Gamma,A}}_{\hvar
          B}}{\interp{A}_{\hvar B}}
    \]
  \item Given a ps-context \(\Gamma\), a type \(\typerule{\Gamma}{A}\) and a
    substitution \(\subrule{\Gamma}{\gamma}{(\cvar x : \obj)}\) in \catt, by an
    argument similar to that of the case of context, by
    Lemma~\ref{lemma:translation-gsett} applied to the term
    \(\termrule{(\cvar x : \obj)}{\cvar x}{\obj}\), we have a derivation of
    \(\termrule{\interp{(\cvar x : \obj)}_{\hvar B}}{\hvar x}{\El(\hvar B)}\).
    By Lemma~\ref{lemma:iterated-refl}, we get a derivation of the judgement
    \[
      \termrule{\interp{(\cvar x : \obj)}_{\hvar B}}{\refl^{\dim A+1}_{\hvar
          B,\hvar x}}{\Id^{\dim A+1}_{\hvar B,\hvar x}}.
    \]
    Then, applying~\eqref{eq:cut-admissibility} provides a derivation of
    \(\typerule{(\cvar x : \obj)}{A[\gamma]}\), and thus
    Lemma~~\ref{lemma:translation-init} then shows
    \(\interp{A[\gamma]}_{\gamma} \equiv \Id^{\dim A+1}_{\hvar B,\hvar x}\),
    which by~\eqref{eq:congruence} gives a derivation of the judgement
  \[
    \termrule{\interp{(\cvar x : \obj)}_{\hvar B}}{\refl^{\dim A+1}_{\hvar
        B,\hvar x}} {\interp{A[\gamma]}_{\hvar B}}.
  \]
  which is syntactically equal to
  \[
    \termrule{\interp{(\cvar x : \obj)}_{\hvar B}}{\elim_{\hvar B}((\cvar
      x:\obj),A,\Gamma,\gamma)} {\interp{A[\gamma]}_{\hvar B}}.
  \]
\item Given a ps-context \(\Gamma\), a type \(\typerule{\Gamma}{A}\) a
  ps-context and a substitution respectively of the following forms
  \begin{align*}
    \Delta' & = (\Delta,\cvar z : C, \cvar f : \arr[C]{\cvar y}{\cvar z})\\
    \delta' &= \sub{\delta,\cvar z \mapsto v, \cvar f \mapsto w}
  \end{align*}
  satisfying \(\subrule{\Gamma}{\delta'}{\Delta'}\), recall that we have
  \[
    \elim_{\hvar B}((\Delta,\cvar y:C,f:\arr[C]{\cvar x}{\cvar y}), A,\Gamma,\gamma)
    \defeq
    \J(
    {\interp{C}_{\hvar B}},
    \hvar y,
    {P(\hvar {y'},\hvar{z'},\hvar{f'})},
    p)\
    \hvar z\
    \hvar f
\]
  where \(P\) and \(p\) are respectively defined as follows
  \begin{align*}
    P(\hvar{y'},\hvar{z'},\hvar{f'})
    & \defeq \interp{A}_{\hvar
      B} [\gamma][\sub{\id_{\interp{\Delta'}_{\hvar B}},\hvar{y}\mapsto\hvar{y'},\hvar{z}\mapsto\hvar{z'},\hvar{f}\mapsto\hvar{f'}}] \\
    p &\defeq \elim_{\hvar B}(\Delta,A,\Gamma,\gamma\circ
        \sub{\id_{\interp{\Delta}_{\hvar B}}, \hvar z\mapsto \hvar y,
        \hvar f \mapsto \refl_{\interp{C}_{\hvar B},\hvar y}}).
  \end{align*}
  We first show that this application of the \(\J\) rule is well-formed in the
  context \(\interp{\Delta'}_{\hvar B}\), that is that it follows the premises
  of \ruleref{rule:j} and of the function application rule. We first note
  that by Lemma~\ref{lemma:admissibility-ps} we have derivation for the
  following judgements
  \begin{align*}
    \typerule{\Delta'}{C} & & \termrule{\Delta'}{\cvar y}{C}\\
    \termrule{\Delta'}{\cvar z}{C} & & \termrule{\Delta'}{\cvar f}{\arr[C]{\cvar
                                       y}{\cvar z}}
  \end{align*}
  By Lemma~\ref{lemma:translation-gsett}, those provide derivations for each of
  the corresponding judgements
    \begin{align*}
      \typerule{\interp{\Delta'}_{\hvar B}}{\interp{C}_{\hvar B}}
      & & \termrule{\interp{\Delta'}_{\hvar B}}{\hvar y}{\interp{C}_{\hvar B}}\\
      \termrule{\interp{\Delta'}_{\hvar B}}{\hvar z}{\interp{C}_{\hvar B}}
      & & \termrule{\interp{\Delta'}_{\hvar B}}{\hvar f}{\Id_{\interp{C}_{\hvar
          B}}(\hvar y,\hvar z)}
    \end{align*}
    So it suffices to show that \(P\) and \(p\) satisfy the premises of
    \ruleref{rule:j}. First, define the following context and substitution
    in \hott:
    \begin{align*}
      \Theta &\defeq (\interp{\Delta'}_{\hvar{B}},\hvar {y'} : \interp{C}_{\hvar B},
        \hvar{z'} : \interp{C}_{\hvar B},\interp{f'}:\Id_{\interp{C}_{\hvar
               B}}(\hvar y',\hvar z')) \\
      \vartheta & \defeq \sub{\id_{\interp{\Delta'}_{\hvar B}},\hvar y \mapsto \hvar{y'},
      \hvar{z}\mapsto\hvar{z'},\hvar{f}\mapsto\hvar{f'}}.
    \end{align*}
    The judgement \(\subrule{\Theta}{\vartheta}{\interp{\Delta'}}_{\hvar B}\) is
    derivable in \hott, and by the inductive for types,so is the judgement
    \(\typerule{\interp{\Gamma}_{\hvar B}}{\interp{A}_{\hvar B}}\), so
    applying~\eqref{eq:cut-admissibility} twice gives a derivation of the
    judgement
    \[
      \typerule{\Theta} {P(\hvar{y'},\hvar{z'},\hvar{f'})}.
    \]
    Finally, we note that by~\eqref{eq:cut-admissibility}, we have a
    derivation of
    \[
      \subrule{\interp{\Delta}_{\hvar B}}{\gamma\circ
      \sub{\id_{\interp{\Delta}_{\hvar B}},\hvar z \mapsto \hvar y,\hvar f
        \mapsto \refl_{\interp{C}_{\hvar B},\hvar y}}}{\interp{\Gamma}_{\hvar
        B}}.
  \]
  Thus, by the inductive case for \(\elim_{\hvar B}\), decreasing on the first
  argument \(\Delta\), we get a derivation of
  \[
    \termrule{\Delta}{p}{\interp{A}_{\hvar B}[\gamma\circ
      \sub{\id_{\interp{\Delta}_{\hvar B}},\hvar z \mapsto \hvar y, \hvar f
        \mapsto \refl_{\interp{C}_{\hvar B},\hvar y}}]}
  \]
  By functoriality of the action of substitution, we have the equality
  \begin{align*}
    \interp{A}_{\hvar B}[\gamma\circ
    \sub{\id_{\interp{\Delta}_{\hvar B}},\hvar z \mapsto \hvar y, \hvar f
    \mapsto \refl_{\interp{C}_{\hvar B},\hvar y}}]
    &= P(\hvar y, \hvar y, \refl_{\interp{C}_{\hvar B},\hvar y}).
  \end{align*}
  This shows that the application of \(\J\) is valid, and provides with
  \ruleref{rule:j} a derivation of the judgement
  \[
    \termrule{\interp{\Delta'}_{\hvar
        B}}{\elim(\Delta',A,\Gamma,\interp{\gamma}_{\hvar B})}{P(\hvar y,\hvar
      z,\hvar f)}
  \]
  Since we have \(P(\hvar y,\hvar z,\hvar f) = \interp{A}_{\hvar B}[\gamma]\),
  this in facts gives a derivation of
  \[
    \termrule{\interp{\Delta'}_{\hvar
        B}}{\elim(\Delta',A,\Gamma,\interp{\gamma}_{\hvar B})}{\interp{A}_{\hvar
      B}[\gamma]}
  \]
\end{itemize}
\end{proof}

\section{The Eckmann-Hilton cell in CaTT}
\label{app:eh-catt}
\begin{figure}
  \centering
\begin{verbatim}
coh unitl (x(f)y) : comp (id _) f -> f
coh unit (x) : comp (id x) (id x) -> id x
coh lsimp (x) : (unitl (id x)) -> unit x
coh Ilsimp (x) : I (unitl (id x)) -> I (unit x)
coh exch (x(f(a)g)y(h(b)k)z) :
  comp (comp _ [b]) (id (comp f k)) (comp [a] _) -> comp [a] [b]

coh eh1 (x(f(a)g(b)h)y) :
comp a b -> comp (I (unitl f))
                 (comp (comp _ [a])
                       (comp (unitl g) (I (op { 1 } (unitl g))))
                       (comp [b] _))
                 (op { 1 } (unitl h))

let eh2 (x : *) (a : id x -> id x) (b : id x -> id x) =
comp [Ilsimp _]
     [comp (comp _
                 [comp
                   (comp [lsimp _] [op { 1 } (Ilsimp _)])
                   (U (unit _))]
                 _)
           (exch b a)]
     [op { 1 } (lsimp _)]

let eh (x : *) (a : id x -> id x) (b : id x -> id x) =
comp (eh1 a b)
     (eh2 a b)
     (I (op { 1 } (eh2 b a)))
     (I (op { 1 } (eh1 b a)))

\end{verbatim}
  \caption{Definition of the Eckmann-Hilton cell in \catt}
  \label{fig:eh}
\end{figure}

\end{document}